\documentclass[11pt,journal,onecolumn]{IEEEtran}
\usepackage{graphicx} 
\graphicspath{{figures/}}
\usepackage{amsthm}
\usepackage{xcolor}
\newtheorem{theorem}{Theorem}
\newtheorem*{theorem*}{Theorem}

\newtheorem{lemma}{Lemma}

\newtheorem{definition}{Definition}
\newtheorem{corollary}{Corollary}
\newtheorem{assumption}{Assumption}
\newtheorem{remark}{Remark}
\usepackage{cite}
\usepackage{amsmath,amssymb,amsfonts}
\usepackage{algorithmicx}
\usepackage{textcomp}
\usepackage[hypertexnames=false]{hyperref}
\usepackage{tikz}
\usepackage{main}
\usepackage{algorithm}
\usepackage{algpseudocode}
\usepackage{enumitem}
\usetikzlibrary{arrows.meta,positioning,fit,calc,decorations.pathmorphing}

\newcommand{\R}{\mathbb{R}}
\newcommand{\I}{\mathsf{I}}
\newcommand{\op}{\mathrm{op}}

\newcommand{\safe}{\mathrm{safe}}

\title{Quantized Online LQR}
\author{Barron Han, Victoria Kostina, and Babak Hassibi%
\thanks{Barron Han, Victoria Kostina, and Babak Hassibi are with the Department of Electrical Engineering and Computing and Mathematical Sciences, California Institute of Technology, Pasadena, CA 91125 USA (e-mail: \{bshan, vkostina, hassibi\}@caltech.edu).}%
}

\begin{document}
\maketitle

\begin{abstract}
We study online linear--quadratic regulation (LQR) with unknown dynamics under communication rate constraints. Classical networked control quantizes the plant state at every time step, requiring $\mathcal{O}(T)$ total bits while injecting persistent quantization noise that limits control performance. We consider a setting where the plant observes its state locally and can estimate system dynamics via ordinary least squares, while a remote controller possesses knowledge of the control cost. Rather than quantizing the raw state, the plant transmits learned dynamics estimates over a rate-limited uplink, and the controller returns the optimal control policy so that the plant can compute actions locally using its superior state knowledge. We first prove a fundamental information-theoretic lower bound: any scheme achieving $\mathcal{O}(T^\alpha)$ regret for $\alpha \in [1/2,1)$ compared to the optimal infinite horizon LQR controller that knows the true system dynamics must transmit at least $\Omega(\log T)$ bits. We then design the \textbf{Quantized Certainty Equivalent (QCE-LQR)} algorithm, which matches this bound. The resulting regret bound contains inflation factors $Q_{\mathrm{slow}}(\varrho)$ and $Q_{\mathrm{fast}}(\varrho)$ that vanish as the codebook resolution increases, smoothly recovering the unquantized baseline regret. Numerical experiments on four benchmark systems---from a scalar unstable plant to a 24-parameter Boeing 747 lateral model---confirm that a variant of QCE-LQR achieves regret comparable to an unquantized certainty equivalent controller over a horizon of $T=10{,}000$ steps.
\end{abstract}


\section{Introduction}

The online LQR problem is a canonical model for analyzing the regret guarantees of data-driven adaptive control algorithms \cite{simchowitz20a, mania2019}. Practical deployments face tight limits on communication bandwidth and onboard computation, motivating theoretical results that quantify performance--communication tradeoffs and algorithms that provably operate under rate constraints.

The online LQR problem without communication constraints has been well studied \cite{mania2019, simchowitz20a, kargin}. Over a horizon of $T$, the optimal regret relative to the infinite-horizon LQR controller with known dynamics is $\tilde{\mathcal{O}}(\sqrt{T})$ (where $\tilde{\mathcal{O}}$ suppresses logarithmic factors). This rate is achieved by standard $\varepsilon$-greedy exploration schemes when a stabilizing safety controller is available \cite{mania2019}, and remains achievable without such a controller via more sophisticated strategies \cite{kargin}.

The offline LQR problem with communication constraints has also been of recent interest. The fundamental tradeoff between LQR cost and communication rate is characterized in \cite{kostina2019}. Joint source-channel codes for LQR over noisy communication channels have been studied in \cite{tatikondacontrol, barronisit, sahaimitter}. Crucially, these classical networked control schemes quantize the \emph{state} $x_t$ at every time step, requiring $\mathcal{O}(T)$ total bits over a horizon of $T$. Moreover, continuous state-quantization noise is injected into the feedback loop at each step, which fundamentally limits the achievable regret when the dynamics are unknown.

These works also assume the linearized dynamics are known a priori, which is often unrealistic. Most real-world systems are inherently non-linear, and their corresponding linear models are obtained by computing the local Jacobian around a nominal operating point. As a system transitions between different operating regimes---due to changes in velocity, varying payload distributions, or environmental wear---this local approximation drifts, causing the effective system parameters to vary with the changing operating conditions. This motivates a data-driven approach that learns the linearized dynamics online. Furthermore, unlike \cite{kostina2019, tatikondacontrol, barronisit, sahaimitter}, we assume the actuator has direct access to the system state and can apply control inputs without quantization, so that communication is required only for transmitting the control \emph{policy}.

An alternative model-free problem setting considers designing the optimal control policy without explicitly learning the system dynamics. For this problem, \cite{mitra2024} devise the Adaptively Quantized Gradient Descent (AQGD) algorithm based on \cite{fazel18a} to learn LQR controllers over rate-limited channels. Their main result shows that, above a finite bit-rate threshold, AQGD converges exponentially fast to the globally optimal policy with no degradation of the convergence rate relative to the unquantized setting. However, their framework studies convergence to the optimal \emph{policy} rather than \emph{online regret}, and requires the agent to observe multiple independent trajectories of the system to form policy gradient estimates.

In this work, we study online LQR with control policy quantization, motivated by a fundamental \emph{information asymmetry} between the system and the controller: the system observes its state locally and can form accurate dynamics estimates via ordinary least squares (OLS), while the controller possesses knowledge of the control cost. To address this asymmetry, the system transmits its dynamics estimates to the controller, which returns the optimal control \emph{policy} $K_t$---not the control action $u_t$---computed from those estimates and the cost objective. The system then applies $u_t = K_t x_t$ locally, exploiting its superior knowledge of the state, $x_t$. Unlike classical networked control \cite{sahaimitter, kostina2019}, the controller transmits the policy rather than the action and therefore does not need access to the system state. This extends naturally to tracking problems, where the controller transmits the target state alongside the policy gain, yielding an affine control law that still requires no knowledge of the current state. In the offline setting where the linear system dynamics are known, this is trivial: the controller computes the optimal gain $K$ once and transmits it at $t=1$, requiring only $\mathcal{O}(1)$ bits. In the online setting, the controller cannot compute the optimal $K$ since the dynamics are unknown, and exploration noise must be injected to accelerate learning \cite{mania2019, simchowitz20a}. The goal is to recover the optimal $\tilde{\mathcal{O}}(\sqrt{T})$ regret scaling.

A naive approach of uniformly quantizing the OLS dynamics estimates fails to preserve the optimal $\tilde{\mathcal{O}}(\sqrt{T})$ regret for two reasons. \emph{First}, the OLS estimation error of the linearized system is highly anisotropic: different parameter subspaces converge at different polynomial rates \cite{simchowitz20a}, so a single-scale quantizer is forced to track the slowest rate and incurs unnecessary dimensional overhead in the dominant regret term. \emph{Second}, before these asymptotic rates take effect, the transient OLS error cannot be tightly bounded, so any fixed-radius codebook risks overflow. Both challenges necessitate an adaptive, mixed-scale quantization scheme.

In this paper, we study online LQR with rate-limited feedback and establish sharp limits on the communication needed to retain optimal regret. We first prove a fundamental information-theoretic lower bound: any scheme achieving $\mathcal{O}(T^\alpha)$ regret for $\alpha \in [1/2,1)$ must transmit at least $\Omega(\log T)$ bits, even if the true dynamics are known by the system. We then construct the \textbf{Quantized Certainty Equivalent (QCE-LQR)} algorithm, which matches this lower bound in its dependence on the horizon $T$. QCE-LQR is a rate-limited variant of the $\varepsilon$-greedy exploration certainty equivalent policy from \cite{simchowitz20a}, modified to compress the evolving sequence of dynamics estimates into an information-theoretically optimal $\mathcal{O}(\log T)$ bit stream without degrading the $\tilde{\mathcal{O}}(\sqrt{T})$ regret scaling.

We present the following informal theorem that combines our main results in Theorems \ref{thm:main_converse} and \ref{thm:main_ach}.

\begin{theorem*}[Informal]
    In the quantized online LQR setting for horizon $T$, a total communication budget $B(T)=\Theta(\log T)$ is necessary and sufficient to achieve $\tilde{\mathcal{O}}(\sqrt{T})$ regret relative to the infinite-horizon LQR controller with known dynamics.
\end{theorem*}

\section{Setting}\label{sec:setting}

\textbf{Notation:} We write $\rho(\cdot)$ for the spectral radius of a matrix, $\|\cdot\|_{\mathrm{op}}$ for the operator (spectral) norm, and $\|\cdot\|_{\mathrm{F}}$ for the Frobenius norm. We use $\lesssim$ and $\gtrsim$ to denote inequality up to universal constants, and $\asymp$ for equality up to universal constants. All logarithms, entropies, and mutual informations are in base~$2$.

Consider the stochastic linear system
\begin{equation}
    x_{t+1} = A x_t + B u_t + w_t,\qquad
    w_t \stackrel{\text{i.i.d.}}{\sim} \mathcal N(0,\sigma_w^2 I), \quad x_1=0,
    \label{eq:sys}
\end{equation}
where $x_t \in \mathbb{R}^{\dx}$ is the state, $u_t \in \mathbb{R}^{\du}$ is the control input, and the per-step cost is
\begin{equation}\label{eq:cost-fn}
    c(x,u) = x^\top \Rx x + u^\top \Ru u, \qquad \Rx\succ 0,\; \Ru\succ 0.
\end{equation}

\begin{remark}
    The isotropic noise covariance $\sigma_w^2 I$ is without loss of generality when the noise covariance is known. Given a system with $w_t \sim \mathcal{N}(0, \Sigma)$ for a known $\Sigma \succ 0$, let $L$ denote its Cholesky factor ($\Sigma = LL^\top$). The change of coordinates $\bar{x}_t = L^{-1} x_t$ yields $\bar{x}_{t+1} = \bar{A}\,\bar{x}_t + \bar{B}\,u_t + \bar{w}_t$ with $\bar{A} = L^{-1}AL$, $\bar{B} = L^{-1}B$, $\bar{w}_t \sim \mathcal{N}(0, I)$, and cost matrices $\bar{R}_x = L^\top R_x L \succ 0$, $\bar{R}_u = R_u$.
\end{remark}
For a stationary controller \(K\), define the infinite-horizon average cost
\begin{equation} \label{cost}
    J(A,B,K) := \lim_{T\to\infty}\frac{1}{T}\,\E  \Big[\sum_{t=1}^T c(x_t,Kx_t)\Big],
    \qquad
    \Jstar(A,B):=\inf_K J(A,B,K),
\end{equation}
and define the finite-horizon regret of a (possibly time-varying, randomized) policy \(\pi=\{u_t\}\) by
\begin{equation}
    \mathrm{Regret}_T := \sum_{t=1}^T c(x_t,u_t)  -  T\,\Jstar(A,B).
    \label{eq:regret}
\end{equation}
Simchowitz and Foster showed that without a rate limit, the optimal regret scales as $\Theta(\du^2 \dx\sqrt{T})$ \cite{simchowitz20a}.

\begin{assumption} \label{asm:pre}
    We assume that the system $(A,B)$ is stabilizable and a stabilizing $K_0$, $\rho(A - BK_0) < 1$, is known by both the controller and system upfront.
\end{assumption}

In our setting, communication is over a noiseless but \emph{asymmetric} bidirectional channel, consistent with standard IoT and edge architectures where the battery-powered plant is power-limited on its uplink while the grid-powered controller has ample downlink capacity. The uplink (plant $\to$ controller) is rate-limited: let \(\Pi_t\) denote the uplink transcript available to the controller when the time-$t$ policy is selected, and let \(B(t)\) denote its cumulative bit length, with total \(B(T)\) by horizon \(T\). This transcript is generated causally from plant-side information available before the action at time $t$ is applied. The downlink (controller $\to$ plant) is unconstrained, allowing the controller to transmit the full-precision policy matrix. We assume the class of linear policies $u_t = K_t x_t$. At time \(t\), the control policy $K_t$ must be measurable with respect to \(\Pi_t\) (and the controller's local randomness).

While an unconstrained downlink theoretically allows the controller to transmit the cost matrices $(\Rx, \Ru)$, enabling the plant to solve the Discrete Algebraic Riccati Equation (DARE) locally to obtain $K_\infty(A,B)$, transmitting only the policy $K_t$ reflects key operational realities in cloud-edge systems. 
First, in fleet deployments, cost matrices encode proprietary strategic priorities; transmitting only $K_t$ obfuscates these costs, as inverse-LQR is ill-posed. Second, edge devices lack the computational capability for iterative DARE solvers (e.g., Schur decompositions), whereas OLS estimation requires only basic matrix operations. Finally, in decentralized systems like energy grids, the effective cost reflects dynamic global states. Streaming this full context is bandwidth-prohibitive, so the cloud must evaluate the global objective and transmit only the compiled local policy $K_t$.

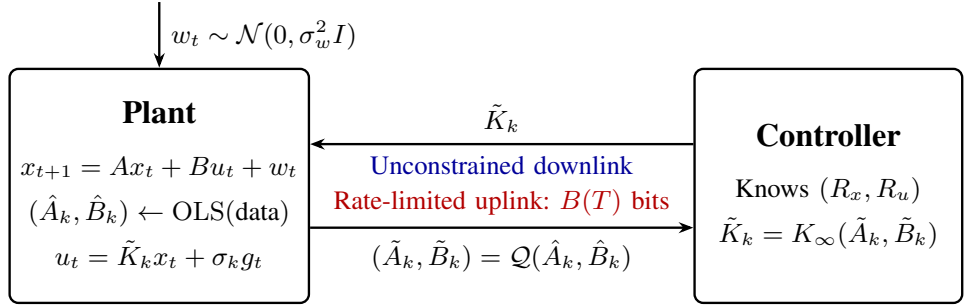
\begin{figure}[t]
    \centering
    \resizebox{0.7\columnwidth}{!}{%
    \begin{tikzpicture}[
        block/.style={draw, rounded corners=3pt, minimum width=3.2cm, minimum height=2.8cm, align=center, thick},
        arr/.style={-{Stealth[length=5pt]}, thick},
        lbl/.style={font=\footnotesize, align=center},
        eq/.style={font=\footnotesize\itshape, text=black!70}
    ]
    \node[block] (plant) at (0,0) {
        \textbf{Plant}\\[4pt]
        {\footnotesize $x_{t+1} = Ax_t + Bu_t + w_t$}\\[2pt]
        {\footnotesize $(\hat A_k, \hat B_k) \gets \mathrm{OLS}(\text{data})$}\\[2pt]
        {\footnotesize $u_t = \tilde K_k x_t + \sigma_k g_t$}
    };

    \node[block] (ctrl) at (8,0) {
        \textbf{Controller}\\[4pt]
        {\footnotesize Knows $(R_x, R_u)$}\\[2pt]
        {\footnotesize $\tilde K_k = K_\infty(\tilde A_k, \tilde B_k)$}
    };

    \draw[arr] ([yshift=-0.5cm]plant.east) -- node[below, lbl] {$(\tilde A_k, \tilde B_k) = \mathcal{Q}(\hat A_k, \hat B_k)$} node[above, lbl, text=red!70!black] {Rate-limited uplink: $B(T)$ bits} ([yshift=-0.5cm]ctrl.west);

    \draw[arr] ([yshift=0.5cm]ctrl.west) -- node[above, lbl] {$\tilde K_k$} node[below, lbl, text=blue!60!black] {Unconstrained downlink} ([yshift=0.5cm]plant.east);

    \draw[arr] (0,2.2) -- node[right, lbl] {$w_t \sim \mathcal{N}(0, \sigma_w^2 I)$} (plant.north);

    \end{tikzpicture}%
    }
    \caption{System model. Communication occurs sparsely at block-wise \emph{epochs} indexed by $k$ (introduced in Section \ref{sec:algorithm-overview}, below). At each epoch, the plant computes OLS estimates $(\hat A_k, \hat B_k)$ of the unknown dynamics and transmits a quantized version $(\tilde A_k, \tilde B_k)$. The controller, which knows the cost matrices $(R_x, R_u)$, solves the DARE to compute the certainty-equivalent policy $\tilde K_k = K_\infty(\tilde A_k, \tilde B_k)$ and transmits it to the plant without rate constraints. At every time step $t$, the plant applies the current epoch's policy $u_t = \tilde K_k x_t + \sigma_k g_t$ with decaying exploration noise.}
    \label{fig:system}
\end{figure}

We use $t$ to denote the discrete time index, and $k$ to denote the algorithm epoch index, described in Section \ref{sec:algorithm-overview}, below. For a stabilizable pair $(A,B)$, let $P_\infty(A,B)$ denote the unique stabilizing solution to the discrete algebraic Riccati equation (DARE)
\begin{equation}\label{eq:dare}
    P = A^\top P A + \Rx - A^\top P B\,(R_u + B^\top P B)^{-1} B^\top P A,
\end{equation}
and let $K_\infty(A,B)$ be the corresponding optimal stationary gain,
\begin{equation}\label{eq:Kstar}
    K_\infty(A,B) := -(R_u + B^\top P_\infty(A,B)\, B)^{-1} B^\top P_\infty(A,B)\, A.
\end{equation}
Denote the optimal closed-loop matrix by
\begin{equation}\label{eq:Acl-star}
    A_{\mathrm{cl},\star} := A + B K_\infty(A,B).
\end{equation}
For a Schur-stable matrix $M$ and a positive semidefinite matrix $Q$, let $\mathsf{dlyap}[M, Q]$ denote the unique solution $X \succeq 0$ to the discrete Lyapunov equation $X = M^\top X M + Q$. When $Q = I$, we abbreviate $\mathsf{dlyap}[M] := \mathsf{dlyap}[M, I]$.

Define:
\begin{align}
    \Jstar &:= \sigma_w^2 \Tr(P_\infty(A,B)), & &\text{(optimal cost),} \label{eq:Jstar_tr} \\
    J_0 &:= J(\Astar,\Bstar,K_0), & &\text{(initial controller cost),} \label{eq:J0} \\
    \Pzero &:= J_0 / \dx, & &\text{(normalized initial cost),} \label{eq:P0} \\
    \PsiB &:= \max\{1,\, \opnorm{\Bstar}\}, & &\text{(controllability proxy).} \label{eq:PsiB}
\end{align}
When the arguments $(A,B)$ are omitted from these functions (e.g., $P_\infty$), they correspond to the true system parameters.
Following \cite{simchowitz20a}, define the \emph{safe constant}
\begin{equation}\label{eq:Csafe-def}
    \Csafe(A,B) := 54\,\opnorm{P_\infty(A,B)}^5,
\end{equation}
which governs the radius within which certainty-equivalent control remains stabilizing, in Section~\ref{sec:quantized-perturbations} below. We use $C_0$ throughout to denote a universal constant arising from the OLS self-normalized bound in Section~\ref{sec:regret-safe-rounds} below.

\section{Converse: Necessity of \texorpdfstring{$\Omega(\log T)$}{Omega(log T)} Bits}
\label{sec:converse}

In this section, we establish a fundamental information-theoretic limit on the communication required for near-optimal adaptive control. We prove that any scheme achieving sub-linear regret must transmit at least $\Omega(\log T)$ bits from the plant to the controller.

Fix $r>0$ and define the class of stabilizable instances whose optimal gain lies in a Frobenius ball of radius $r$:
\begin{equation}\label{eq:Theta-def}
    \Theta := \bigl\{(A,B) : (A,B) \text{ stabilizable},\ \|K_\infty(A,B)\|_F \le r\bigr\}.
\end{equation}

\begin{theorem}[Fundamental Regret-Communication Trade-off]\label{thm:main_converse}
    Fix $\alpha\in[1/2,1)$ and $C_1>0$. Consider any quantized control scheme of the form $u_t = K_t x_t$, where for each $t$, the random matrix $K_t\in\R^{\du\times\dx}$ is measurable with respect to the controller's received uplink transcript $\Pi_t$ and its local randomness (cf.\ Section~\ref{sec:setting}). If
    \begin{equation}\label{eq:regret-assumption}
        \sup_{\theta \in \Theta} \E[\mathrm{Regret}_T] \le C_1 T^\alpha,
    \end{equation}
    then there exists a constant $C$, depending only on $(\Rx,\Ru,\sigma_w^2,r,\dx,\du,C_1,\alpha)$ and not on $T$, such that
    \begin{equation}
        B(T)  \ge  \frac{\du\dx}{2}\,(1-\alpha)\,\log_2 T  -  C.
        \label{eq:bits-lb}
    \end{equation}
\end{theorem}

\begin{corollary}[Optimal Regret Requires $\Omega(\log T)$ Bits]\label{cor:sqrt-t}
    To achieve the optimal unquantized regret scaling of $\E[\mathrm{Regret}_T] = \mathcal{O}(\sqrt{T})$ corresponding to $\alpha = 1/2$ in \eqref{eq:regret-assumption}, the total communication budget must scale as:
    \begin{equation}
        B(T) \ge \frac{\du\dx}{4}\,\log_2 T - C.
        \label{eq:bits-lb-sqrt}
    \end{equation}
\end{corollary}

\begin{proof}[Proof sketch]
    The full proof is in Appendix~\ref{app:converse}. We construct a concrete hard subclass $\Theta_{\mathrm{hard}} \subset \Theta$ parametrized by a cube of gains $K \in [-a,a]^{\du \times \dx}$, where $a := r/\sqrt{\du\dx}$ is chosen so that every $K$ in the cube satisfies $\|K\|_F \le r$ and hence $\Theta_{\mathrm{hard}} \subset \Theta$. Each $K$ indexes a stabilizable instance $(A_K, B_K)$ for which $K_\infty(A_K,B_K) = K$, and all instances share the common optimal cost $J_\infty(A_K,B_K) = \sigma_w^2 \Tr(P)$. A regret identity then decomposes horizon-$T$ regret into a sum of per-step excess costs, each lower bounded by a constant times $\E\|K_t - K\|_F^2$ via the process-noise covariance floor. Averaging over the first half of the horizon produces a single estimator $\widehat K$ of $K$ satisfying $\E\|\widehat K - K\|_F^2 \le C_{\mathrm{est}}\,T^{\alpha-1}$. Finally, placing a uniform prior on the cube and applying the data-processing inequality bounds the mutual information by the transcript length $B(T)$. We lower bound the distortion in $K$ using a Gaussian reduction. Combining these yields~\eqref{eq:bits-lb}.
\end{proof}

\section{Achievability: \texorpdfstring{$\tilde{\mathcal{O}}(\sqrt T)$}{O\textasciitilde(sqrt T)} Regret with \texorpdfstring{$\mathcal{O}(\log T)$}{\mathcal{O}(log T)} Bits}
\label{sec:achievability}

In this section, we present the Quantized Certainty Equivalent LQR (QCE-LQR) algorithm, which builds upon the unquantized adaptive control scheme of \cite{simchowitz20a}. By integrating our novel quantization protocols, we show how the plant can compress its local dynamics estimates into $\mathcal{O}(\log T)$ total bits without degrading the $\tilde{\mathcal{O}}(\sqrt{T})$ regret scaling.

Define dimensional constants
\begin{equation}\label{eq:ds-def}
    d := \dx + \du, \qquad d_s := \dx^2 + \dx\du,
\end{equation}
where $d_s$ is the total number of unknown system parameters. For a codebook resolution $\varrho \in (0,1/\sqrt{2})$, define
\begin{equation}\label{eq:Crho-def}
    \Crho := \frac{1}{1-\varrho\sqrt{2}},
\end{equation}
and the quantization inflation factors
\begin{align}
    Q_{\mathrm{slow}}(\varrho) &= \varrho^2 \Crho^2 \left( \frac{1+2^{1/4}}{1-\varrho\,2^{1/4}} \right)^{\!2} \cdot C_0\,(1.0835)^4, \label{eq:Qslow} \\
    Q_{\mathrm{fast}}(\varrho) &= \varrho^2 \Crho^2 \left( \frac{1+\sqrt{2}}{1-\varrho\sqrt{2}} \right)^{\!2} \cdot C_0\,(1.0835)^6. \label{eq:Qfast}
\end{align}
Finally, define the OLS mixing time and the pre-safe duration bound (Lemma~\ref{lem:presaferegret} in Section \ref{sec:presafe-regret}, below):
\begin{align}
    \tauls &:= d \max\{1, \du/\dx\} \left(\Pstarnorm^3 \Pzero + \Pstarnorm^{11} \PsiB^6\right) \log \frac{d\Pstarnorm}{\delta}, \label{eq:tauls-def} \\
    \tausafe &:= d\,\Pstarnorm^{10}\,(1+\opnorm{K_0}^2)\,\log\frac{\PsiB^2 \Jzero}{\delta}. \label{eq:tausafe-def}
\end{align}
For the communication bound, let $\beta := \varrho\sqrt{2}$ and
\begin{equation}\label{eq:m-infty-def}
    m_\infty := \frac{2 - \beta}{1 - \beta}, \qquad b_\varrho := 2\left\lfloor \log_2 \lceil m_\infty \rceil \right\rfloor + 1.
\end{equation}
Let $\theta_{\max}$ denote the largest absolute entry of the system matrices,
\begin{equation}\label{eq:theta-max-def}
    \theta_{\max} := \max_{i,j}\{|A_{ij}|, |B_{ij}|\},
\end{equation}
and define the one-time initialization cost
\begin{equation}\label{eq:Binit-def}
    B_{\mathrm{init}} := \mathcal{O}\!\left(d_s \log\!\left(1 + \theta_{\max}\, \Csafe(\Astar, \Bstar)\, \sqrt{d_s}\right)\right).
\end{equation}

\begin{theorem}[Main Achievability Result]
    \label{thm:main_ach}
    Under the setting of Section~\ref{sec:converse}, for any $\delta \in (0,\frac{1}{T})$, the QCE-LQR algorithm (Algorithm~\ref{alg:quantized_ce_lqr_proj}) satisfies the following with probability at least $1-\delta$.

    \medskip
    \noindent\textbf{(a) Total Communication.} The total number of bits transmitted by horizon $T$ is bounded by
    \begin{equation}
        \label{eq:total-bits}
        B(T) \le \left[ d_s \log_2\!\left(1+\frac{2}{\varrho}\right) + b_\varrho \right] \log_2 T + B_{\mathrm{init}} + \mathcal{O}\!\left(\log \tausafe + \log^2 \tauls\right) = \mathcal{O}\big( d_s \log T \big).
    \end{equation}
    \medskip
    \noindent\textbf{(b) Total Regret.} The regret satisfies
    \begin{align}
        \mathrm{Regret}_T \lesssim \;& \sqrt{T}\left( \du\sqrt{\dx}\,\PsiB^3\,\Pstarnorm^{11/2}\,\sqrt{\log\tfrac{\Pstarnorm}{\delta}} + \sqrt{d\log\tfrac{1}{\delta}}\,\Pstarnorm^4 + Q_{\mathrm{slow}}(\varrho)\,\frac{\dx\du}{\sigmain^2}\,\Pstarnorm^{10}\,\log\tfrac{1}{\delta} \right) \label{eq:regret-dominant} \\
        & + \log T \cdot Q_{\mathrm{fast}}(\varrho)\,\dx^2\,\Pstarnorm^{11}\,\log^2\!\tfrac{1}{\delta} \nonumber \\
        & +  d^2 \PsiB^2 \Pzero \Pstarnorm^{10} (1+\opnorm{K_0}^2) \log \tfrac{\PsiB^2 \Jzero}{\delta}\, \log \tfrac{1}{\delta} \nonumber \\
        & + \sqrt{d\dx}\,\PsiB^2\,\Pstarnorm^{17/2}\,\sqrt{\log\tfrac{\Pstarnorm}{\delta}}\,\log^2\!\tfrac{1}{\delta} + \PsiB\,\Jzero\,\Pstarnorm^3\, \tauls \log\!\tfrac{1}{\delta} \nonumber \\
        & + \dx^2\, \Pstarnorm^3\, \tauls \log^2\!\tfrac{1}{\delta}. \nonumber
    \end{align}
\end{theorem}

Theorem~\ref{thm:main_ach} shows that $\mathrm{Regret}_T = \tilde{\mathcal{O}}(\sqrt{T})$ is achievable using only $\mathcal{O}(\log T)$ bits. Part~(a) establishes that the total uplink cost grows as $[d_s \log_2(1 + 2/\varrho) + b_\varrho] \cdot \log_2 T$ plus a one-time initialization overhead $B_{\mathrm{init}}$ that depends on the system parameters but not on $T$; the detailed accounting is in Appendix~\ref{app:comm-bound}. Part~(b) gives a regret bound with $\sigmain^2$ from~\eqref{eq:sigmain-def} (see Appendix~\ref{app:achievability}). The dominant $\sqrt{T}$ coefficient is $\du\sqrt{\dx}\,\PsiB^3\,\Pstarnorm^{11/2}\,\sqrt{\log(\Pstarnorm/\delta)}$, yielding the optimal $\tilde{\mathcal{O}}(\sqrt{\du^2\dx\,T})$ dimension scaling of~\cite{simchowitz20a}. The quantization contributes the slow-rate term $Q_{\mathrm{slow}}(\varrho)\,\frac{\dx\du}{\sigmain^2}\,\Pstarnorm^{10}\log(1/\delta)$ at the $\sqrt{T}$ scale and the fast-rate term $Q_{\mathrm{fast}}(\varrho)\,\dx^2\,\Pstarnorm^{11}\log^2(1/\delta)$ at the $\log T$ scale, and both inflation factors vanish as $\varrho \to 0$.

\subsection{Algorithm Overview}
\label{sec:algorithm-overview}

Our main algorithm, Quantized Certainty Equivalent LQR (\textbf{QCE-LQR}) (Algorithm~\ref{alg:quantized_ce_lqr_proj}), is a rate-limited adaptation of the $\varepsilon$-greedy scheme from \cite{simchowitz20a}. The algorithm takes as input a stabilizing controller $K_0$ and proceeds in doubling epochs $k$ of length $\tau_k = 2^k$. During an initial burn-in period, the plant uses $K_0$ with Gaussian excitation to explore the state space. Once the plant's local, unquantized OLS estimate achieves sufficient statistical reliability at epoch $\ksafe$, it transmits this initial estimate to the controller using an absolute Elias Gamma encoding. This establishes a shared, reliable baseline model $(\tilde A_{\ksafe}, \tilde B_{\ksafe})$ between the plant and the controller. From this baseline, both tracking sides natively compute a known \emph{safe set} $\mathcal{B}_{\mathrm{safe}}$ designed to guarantee closed-loop stability for any contained parameter.

In the post-safe phase ($k > \ksafe$), the plant relies on differential quantization. At the end of each epoch $k$, it computes an updated OLS estimate
\begin{equation}\label{eq:theta-hat}
    \hat\theta_k := \mathrm{vec}(\hat A_k,\hat B_k) \in \mathbb{R}^{d_s}
\end{equation}
and sends a \emph{quantized innovations update} relative to the previously shared estimate. Writing $\tilde\theta_k := \mathrm{vec}(\tilde A_k, \tilde B_k)$ for the shared decoded parameter, both sides maintain it via
\begin{equation}\label{eq:diff-update}
    \tilde\theta_k \;=\; \tilde\theta_{k-1} + s_k c_{q_k},
\end{equation}
where $q_k$ is the transmitted codeword index corresponding to a vector $c_{q_k}$ from a fixed codebook $\mathcal C$ (a $\varrho$-net of the unit ball in $\mathbb R^{d_s}$), and $s_k$ is a scalar radius. Given the sequence of shared parameters, the controller projects onto the known safe set:
\begin{equation}\label{eq:safe-proj-overview}
    \check\theta_k := \mathrm{vec}(\check A_k, \check B_k) = \Pi_{\mathcal{B}_{\mathrm{safe}}}(\tilde{\theta}_k),
\end{equation}
which ensures closed-loop stability with high probability. It then synthesizes the certainty-equivalent controller
\begin{equation}\label{eq:Ktil-overview}
    \tilde K_k := K_\infty(\check A_k, \check B_k)
\end{equation}
and transmits it back to the plant over the unconstrained downlink. The plant applies this policy with decaying exploration noise:
\begin{equation}\label{eq:control-input}
    u_t = \tilde K_k x_t + \sigma_k g_t, \qquad \sigma_k^2 := \min\{1,\,\sigmain^2\,\tau_k^{-1/2}\},
\end{equation}
where $g_t \overset{\mathrm{i.i.d.}}{\sim} \mathcal{N}(0, I_{\du})$ is isotropic Gaussian excitation and $\sigmain^2 > 0$ is an exploration constant computed at epoch $\ksafe$, specified in Lemma~\ref{lem:51} (Section~\ref{sec:quantized-perturbations}, below).

Because the codebook $\mathcal C$ is a fixed, dimension-dependent $\varrho$-net, each epoch requires communicating a constant number of bits $\log_2|\mathcal C|$ for the codeword index. In addition, each epoch carries a small overhead: a $1$-bit safety flag indicating whether the safe trigger has fired, and an Elias-Gamma coded adaptive multiplier $m_k$ (defined in Section~\ref{sec:contraction}, below) that scales the quantization radius to absorb transient estimation errors. By Lemma~\ref{lem:contraction} (Section~\ref{sec:contraction}, below), $m_k$ contracts to $\mathcal{O}(1)$ after a bounded number of epochs, so the total transmission rate per epoch is $\mathcal{O}(1)$ bits asymptotically. Summing over $\asymp \log_2 T$ epochs, this scheme compresses the entire sequence of dynamics estimates into $\mathcal{O}(\log T)$ total bits. The key technical challenge, treated in the subsequent sections, is designing the scale sequence $s_k$ such that quantization errors do not degrade the asymptotic $\tilde{\mathcal{O}}(\sqrt{T})$ regret.

\begin{table}[h]
\centering
\caption{Parameter estimate notation at epoch $k \geq \ksafe$.}
\label{tab:theta-notation}
\begin{tabular}{c l l}
\toprule
\textbf{Symbol} & \textbf{Description} & \textbf{Known to} \\
\midrule
$\hat\theta_k$ & Unquantized OLS estimate~\eqref{eq:theta-hat} & Plant only \\
$\bar\theta_k$ & Plant-side projected estimate~\eqref{eq:plant-proj} & Plant only \\
$\tilde\theta_k$ & Shared decoded (quantized) estimate~\eqref{eq:diff-update} & Both \\
$\check\theta_k$ & Controller-side projected estimate~\eqref{eq:safe-proj-overview} & Both \\
\bottomrule
\end{tabular}
\end{table}
\subsection{Quantization Rule: Adaptive Scale Protocol}

The key challenge in designing the quantization schedule is that the OLS estimation error is \emph{highly anisotropic}. As shown in Corollary~\ref{cor:ols-two-rate} in Section~\ref{sec:regret-safe-rounds}, below, after a bounded period $\tauls$, the error decays as $\tau^{-1/4}$ for the $\dx\du$-dimensional subspace (the ``slow'' rate) but as $\tau^{-1/2}$ for the $\dx^2$-dimensional subspace (the ``fast'' rate). A standard single-scale tracking quantizer would set its radius to $s_k \propto \tau_k^{-1/4}$, forcing the slow rate to dominate and coupling the $\dx^2$ scaling into the dominant $\sqrt{T}$ regret term.

To avoid this, we define the \emph{ideal base schedule}
\begin{equation}
    \label{eq:two-scale-schedule}
    \sbase = \cslow\,\tau_k^{-1/4} + \cfast\,\tau_k^{-1/2},
\end{equation}
where $\cslow$ and $\cfast$ are constants computed at epoch $\ksafe$ in Section~\ref{sec:data-driven}, below. The key design principle is that $\cslow$ and $\cfast$ are calibrated to match the dimension-dependent scaling of the OLS estimation error (Corollary~\ref{cor:ols-two-rate} in Section~\ref{sec:regret-safe-rounds}, below): $\cslow$ isolates the $\dx\du$-dimensional slow-rate component, while $\cfast$ isolates the $\dx^2$-dimensional fast-rate component. Because the quantization error $\varrho s_k$ inherits the same two-scale decay as the underlying estimation error, it adds no order-wise overhead to the total parameter error. While the quantizer remains an isotropic $\varrho$-net of the unit ball, this tighter envelope ensures the $\dx^2$ contribution is quarantined into the lower-order $\log T$ regret term.

After the safe epoch $\ksafe$, the plant projects the OLS estimate onto the safe set,
\begin{equation}\label{eq:plant-proj}
    \bar\theta_k := \Pi_{\mathcal{B}_{\mathrm{safe}}}(\hat\theta_k),
\end{equation}
and computes the innovation
\begin{equation}\label{eq:innovation}
    \Delta_k := \bar\theta_k - \tilde\theta_{k-1}.
\end{equation}
Because the asymptotic OLS bounds only hold after a mixing time $\tauls$ (Corollary~\ref{cor:ols-two-rate} in Section~\ref{sec:regret-safe-rounds}, below), early innovations may exceed $\sbase$. The plant therefore dynamically computes a multiplier
\begin{equation}
    \label{eq:adaptive-mult}
    m_k = \max\!\left\{1,\; \left\lceil \frac{\|\Delta_k\|_2}{\sbase} \right\rceil\right\},
\end{equation}
and sets the actual scale $s_k = m_k \cdot \sbase$. By construction, $\|\Delta_k\|_2 \le s_k$, making quantizer overflow impossible. The multiplier $m_k$ is transmitted via Elias Gamma coding at a cost of $2\lfloor\log_2 m_k\rfloor + 1$ bits. The plant then quantizes using the $\varrho$-net:
\begin{equation}\label{eq:quant-rule}
    q_k \in \arg\min_{c\in\mathcal C}\|\Delta_k - s_k c\|_2, \qquad \|\tilde\theta_k - \bar\theta_k\|_2 \le \varrho s_k.
\end{equation}

\subsection{Two Phases: Before and After the Safe Epoch}

Initially, the controller applies a known stabilizing gain \(K_0\) with Gaussian excitation,
e.g.\ \(u_t = K_0 x_t + g_t\), \(g_t\sim\mathcal N(0,I)\).
This guarantees boundedness and persistent excitation sufficient for OLS to become accurate.

At the end of each epoch \(k\), the plant computes the self-normalized confidence radius
\(\Conf_k\) as in \cite{simchowitz20a}, where $\Csafe$ was defined in~\eqref{eq:Csafe-def}.
The safe epoch is triggered when the following two conditions are both met:
\begin{enumerate}
    \item the empirical covariance satisfies $\mathbf{\Lambda}_k \succeq I$, and
    \item the OLS estimation error is small enough: $\sqrt{\Conf_k} \le \epsilon_{\mathrm{target}}$,
\end{enumerate}
where the target precision is
\begin{equation} \label{eq:eps-target-def}
    \epsilon_{\mathrm{target}} := \frac{1}{9\,\Csafe(\hat A_k,\hat B_k)}.
\end{equation}
The trigger and quantization precision share the same threshold $\epsilon_{\mathrm{target}}$.
This ensures that the decoded center $\tilde\theta_{\ksafe}$ is close enough to $\hat\theta_{\ksafe}$ for the
perturbation lemma (Lemma~\ref{lem:11}) to apply with $\tilde\theta_{\ksafe}$ as the common center
(see Appendix~\ref{app:lemmas} for details).
At this point, the algorithm has gathered sufficient data to switch from the known stabilizing controller $K_0$ to a learned certainty-equivalent controller~\eqref{eq:Ktil-overview}: it sets \(\mathrm{safe}\gets\mathtt{True}\) and declares the corresponding epoch index \(\ksafe\).
Crucially, this trigger depends only on the OLS confidence and not on the quantization scale, so the safe epoch occurs at the same time as in the unquantized algorithm of~\cite{simchowitz20a}, recovering the same regret in the pre-safe period up to constants.

At the safe epoch, the plant performs an initialization via Elias Gamma coding: each of the $d_s$~\eqref{eq:ds-def} coordinates of $\hat\theta_{\ksafe}$ is transmitted as a signed integer on a dyadic grid with step $\Delta = 2^{-E}$, where $E = \lceil\log_2(\sqrt{d_s}/(2\epsilon_{\mathrm{target}}))\rceil$. This yields a shared reconstruction $\tilde\theta_{\ksafe}$ with $\|\tilde\theta_{\ksafe} - \hat\theta_{\ksafe}\|_2 \le \epsilon_{\mathrm{target}}$.

Both plant and controller then set the safe radius from the shared decoded center:
\begin{equation}\label{eq:rsafe-def}
    \rsafe := \frac{1}{3\,\Csafe(\tilde A_{\ksafe}, \tilde B_{\ksafe})}.
\end{equation}
Since the decoded pair $(\tilde A_{\ksafe}, \tilde B_{\ksafe})$ is obtained solely from the transmitted indices,
both sides compute the same $\rsafe$ with no communication cost.
We then form a fixed operator-norm neighborhood (``safe set'')
\begin{equation}\label{eq:Bsafe-def}
    \mathcal B_{\mathrm{safe}}
    := \Big\{(A,B):\ \max\{\|A-\tilde A_{k_{\mathrm{safe}}}\|_{\op},\ \|B-\tilde B_{k_{\mathrm{safe}}}\|_{\op}\}\le \rsafe\Big\},
\end{equation}
and we choose the initial exploration variance
\(\sigmain^2\) using the same explicit formula as in~\cite{simchowitz20a}, but evaluated at the
same shared estimate \((\tilde A_{k_{\mathrm{safe}}},\tilde B_{k_{\mathrm{safe}}})\). This sets the initial scale from which the exploration schedule \(\sigma_k^2\) in~\eqref{eq:control-input} decays.

After \(k_{\mathrm{safe}}\), the controller continues decoding \(\tilde\theta_k\) from the received indices,
but it \emph{does not directly synthesize a controller from \((\tilde A_k,\tilde B_k)\)}.
Instead, it first performs the safe-set projection~\eqref{eq:safe-proj-overview} and synthesizes the certainty-equivalent controller~\eqref{eq:Ktil-overview}, applying the policy~\eqref{eq:control-input} with continual exploration noise.

Let the empirical covariance over epoch $k$ be
\begin{equation}\label{eq:Lambda-k}
    \mathbf{\Lambda}_k := \sum_{t=\tau_{k-1}}^{\tau_k-1}
    \begin{bmatrix}x_t\\u_t\end{bmatrix}
    \begin{bmatrix}x_t\\u_t\end{bmatrix}^{\!\top}.
\end{equation}
Define the confidence scalar
\begin{equation}
    \Conf_k :=
    6\,\lambda_{\min}(\mathbf{\Lambda}_k)^{-1}\Big(
    d\log 5 + \log\Big( \frac{4k^2 \det(3\mathbf{\Lambda}_k/\delta)}{1}\Big)
    \Big),
    \label{eq:Conf}
\end{equation}
with the convention $\Conf_k=\infty$ if $\mathbf{\Lambda}_k\not\succ 0$.

\subsection{Correctness of Quantized Perturbations}
\label{sec:quantized-perturbations}

Let
\begin{equation}\label{eq:Jk-def}
    P_k = P_\infty\left(\tilde K_k, A, B\right), \qquad J_k = \Tr(P_k)
\end{equation}
denote the infinite-horizon LQR cost matrix and cost of the quantized controller at epoch $k \geq k_{\safe}$ on the true system. Define the safe event
\begin{equation}\label{eq:Esafe}
    \mathcal{E}_{\text{safe}} := \left\{ \left\| \left[ \widehat{A}_{k_{\text{safe}}} - A \mid \widehat{B}_{k_{\text{safe}}} - B \right] \right\|_{\mathrm{op}}^2 \leq \text{Conf}_{k_{\text{safe}}} \right\}.
\end{equation}

Write $A_{\mathrm{cl},k} := A + B\tilde K_k$ for the quantized closed-loop matrix at epoch $k$.

The following lemma verifies that the quantized controller $\tilde K_k$ inherits the essential properties of the optimal controller $\Kstar$. It parallels Lemma 5.1 in \cite{simchowitz20a}, which proves the same properties on the certainty equivalent controller generated directly from the OLS estimates. The following lemma contains six bounds that serve distinct roles in the regret analysis:
\emph{(1)}~links per-epoch excess cost to parameter estimation error, driving the main $\sum_k \tau_k(J_k - \Jstar)$ regret term;
\emph{(2)}~ensures the closed-loop Riccati solution does not blow up, which is needed for the self-normalized OLS confidence sets to remain valid;
\emph{(3)}~controls the magnitude of the applied control input, preserving the exploration--exploitation balance required for persistent excitation;
\emph{(4)}~bounds the $\mathcal{H}_\infty$ gain of the closed loop, ensuring process noise is not catastrophically amplified and that state norms remain bounded (Lemma~5.3 in \cite{simchowitz20a});
\emph{(5)}~establishes strict Lyapunov contraction of the quantized closed-loop matrix with respect to the optimal Lyapunov function, guaranteeing closed-loop stability and enabling the telescoping argument in the regret decomposition;
\emph{(6)}~specifies the exploration variance $\sigma_{\mathrm{in}}^2$ that balances fast OLS convergence against exploration cost, achieving the optimal $\tilde{\mathcal{O}}(\sqrt{d_x d_u^2 T})$ dimension scaling.

\begin{lemma} (Correctness of Quantized Perturbations)
    \label{lem:51}
    On the event $\mathcal{E}_{\text{safe}}$, the following bounds hold for all $k \geq k_{\safe}.$

    \begin{enumerate}
        \item $J_k - \Jstar \lesssim \Pstarnorm^8 \left(\|\tilde{A} - A\|_{\text{F}}^2 + \|\tilde{B} - B\|_{\text{F}}^2\right).$

        \item $J_k \lesssim \Jstar, \text{ \textit{and} } \|P_k\|_{\mathrm{op}} \lesssim \Pstarnorm.$

        \item $\|\tilde{K}_k\|_{\mathrm{op}}^2 \leq \frac{21}{20} \|P_\star\|_{\mathrm{op}}.$

        \item $\|A_{\mathrm{cl},k}\|_{\mathcal{H}_\infty} \lesssim \|A_{\mathrm{cl},\star}\|_{\mathcal{H}_\infty} \lesssim \|P_\star\|_{\mathrm{op}}^{3/2}.$

        \item $A_{\text{cl},k}^\top \mathsf{dlyap}[A_{\text{cl},\star}] A_{\text{cl},k} \preceq \left(1 - \frac{1}{2}\|\mathsf{dlyap}[A_{\text{cl},\star}]\|_{\mathrm{op}}^{-1}\right)\mathsf{dlyap}[A_{\text{cl},\star}], \text{ \textit{where} } I \preceq \mathsf{dlyap}[A_{\text{cl},\star}] \preceq P_\star$. 

        \item $\sigmain^2 \asymp \sqrt{\dx}\Pstarnorm^{9/2} \PsiB \sqrt{\log \frac{\Pstarnorm}{\delta}}.$
    \end{enumerate}
\end{lemma}

The proof relies on the following result from \cite{simchowitz20a}.

\begin{lemma}[Thm. 11 from \cite{simchowitz20a}]
    \label{lem:11}
    Let $(A_0, B_0)$ be a stabilizable system. Then, for any pair of systems $(\widehat{A}, \widehat{B}), (A, B)$ satisfying
    \[
        \max \left \{ \| \widehat{A} - A_0 \|_{\mathrm{op}}, \| \widehat{B} - B_0 \|_{\mathrm{op}} \right\} \le \frac{1}{3C_{\text{safe}}(A_0, B_0)}
        \quad \text{and} \quad
        \max \left \{ \| A - A_0 \|_{\mathrm{op}}, \| B - B_0 \|_{\mathrm{op}} \right\} \le \frac{1}{3C_{\text{safe}}(A_0, B_0)},
    \]
    each is stabilizable, and satisfies
    \begin{align}
        \max\{\|A - \widehat{A}\|_{\mathrm{op}},\; \|B - \widehat{B}\|_{\mathrm{op}}\} &\le \frac{1}{C_{\text{safe}}(A, B)}, \label{eq:thm11-safe} \\
        \|P_\infty(A, B)\|_{\mathrm{op}} &\le 1.0835\, \|P_\infty(A_0, B_0)\|_{\mathrm{op}}, \label{eq:thm11-P} \\
        C_\safe(A, B) &\le 1.5\, C_\safe(A_0, B_0). \label{eq:thm11-Csafe}
    \end{align}
\end{lemma}

The proof of Lemma~\ref{lem:51} is given in Appendix~\ref{app:lemmas}.

\subsection{Regret Bound Sketch}
\label{sec:regret-sketch}

The regret decomposes into two phases: pre-safe burn-in, and post-safe tracking.

\textbf{Pre-safe regret.} Before the safe trigger fires, the algorithm executes $K_0$ identically to the unquantized setting of \cite{simchowitz20a}. Because the trigger depends only on the OLS confidence and the quantization error is subsequently constrained to be of similar order, the safe epoch occurs at a bounded stopping time $\tau_{\ksafe} \le \tausafe$~\eqref{eq:tausafe-def}. The accumulated pre-safe regret is
\begin{equation}\label{eq:presafe-regret}
    \begin{split}
        \sum_{t=1}^{\tau_{\ksafe}-1} c(x_t,u_t) &\lesssim d^2 \PsiB^2 \Pzero \Pstarnorm^{10} (1+\opnorm{K_0}^2) \\
        &\quad \times \log\!\tfrac{\PsiB^2 \Jzero}{\delta}\, \log\!\tfrac{1}{\delta},
    \end{split}
\end{equation}
which matches the unquantized pre-safe bound of \cite[Lem. 5.5]{simchowitz20a}.

\textbf{Post-safe regret.} Following \cite[Lem. 5.2]{simchowitz20a}, the post-safe regret decomposes as 
\begin{equation}\label{eq:safe-decomp}
    \begin{split}
        &\sum_{t=\tau_{\ksafe}}^{T} \big(c(x_t,u_t) - \Jstar\big) \\
        &\qquad\lesssim \sum_{k=\ksafe}^{\kfin} \tauk (J_k - \Jstar) \\
        &\qquad\quad + \sqrt{T}\big(\du\sigmain^2\PsiB^2\Pstarnorm + \sqrt{d\log\!\tfrac{1}{\delta}}\,\Pstarnorm^4\big) \\
        &\qquad\quad + \text{lower-order terms}.
    \end{split}
\end{equation}
The perturbation bound gives
\begin{equation}
J_k - \Jstar \lesssim \Pstarnorm^8 \Fnorm{\tilde\theta_k - \theta_\star}^2,
\end{equation}
which splits into an unquantized OLS error and a quantization reconstruction error $\varrho^2 s_k^2$. Using the OLS error bounds and the ideal base schedule~\eqref{eq:two-scale-schedule}, the quantization-induced regret separates as
\begin{equation}\label{eq:quant-regret-split}
    \Pstarnorm^8\,\varrho^2\Crho^2 \big(\cslow^2\sqrt{T} + \cfast^2 \log T\big),
\end{equation}
recovering the inflation factors $Q_{\mathrm{slow}}(\varrho)$ and $Q_{\mathrm{fast}}(\varrho)$. The exploration variance $\sigmain^2$ appears in two competing terms in~\eqref{eq:safe-decomp}: the direct exploration cost $\du\sigmain^2\PsiB^2\Pstarnorm\sqrt{T}$ and, through the OLS estimation cost, a $1/\sigmain^2$ term. Balancing these yields $\sigmain^2 \asymp \sqrt{\dx}\Pstarnorm^{9/2}\PsiB\sqrt{\log(\Pstarnorm/\delta)}$ and recovers the optimal dominant coefficient $\du\sqrt{\dx}\,\PsiB^3\,\Pstarnorm^{11/2}\sqrt{\log(\Pstarnorm/\delta)} \cdot \sqrt{T}$.

\subsection{Regret Analysis for Safe Rounds}
\label{sec:regret-safe-rounds}

Lemma \ref{lem:51} establishes important conditions to analyze the regret incurred by Algorithm~\ref{alg:quantized_ce_lqr_proj} (described in Section~\ref{sec:algorithm-overview}). In particular, it guarantees that the quantized and projected controller $\tilde K_k$ remains stabilizing for the true system $(\Astar,\Bstar)$, with cost and operator-norm bounds comparable to the optimal controller. These conditions allow us to reproduce Lemmas 5.2, 5.3, and 5.4 from \cite{simchowitz20a}.

\begin{lemma}[Regret Decomposition on Safe Rounds, Lem. 5.2 from \cite{simchowitz20a}]
    There is an event $\mathcal{E}_{\mathrm{reg}}$ which holds with probability at least $1 - \frac{\delta}{8}$ such that, on $\mathcal{E}_{\mathrm{reg}} \cap \mathcal{E}_{\mathrm{safe}}$, the following bound holds
    \begin{equation} \label{regredecomp}
        \begin{split}
            &\sum_{t=\tau_{\ksafe}}^{T} \big(x_t^\top \Rx x_t + u_t^\top \Ru u_t - \Jstar\big) \\
            &\qquad\lesssim \sum_{k=\ksafe}^{\kfin} \tauk (J_k - \Jstar) + \log T \max_{k \le \log T} \|x_{\tauk}\|_2^2 \\
            &\qquad\quad + \sqrt{T} \big( \du \sigmain^2 \PsiB^2 \Pstarnorm + \sqrt{d \log(1/\delta)}\,\Pstarnorm^4 \big) \\
            &\qquad\quad + \log^2 \!\tfrac{1}{\delta}\, (1 + \sqrt{d} \sigmain^2 \PsiB^2) \Pstarnorm^4.
        \end{split}
    \end{equation}
\end{lemma}

\begin{lemma}[Lem. 5.3 from \cite{simchowitz20a}]
    There is an event $\mathcal{E}_{\mathrm{bound}}$ which holds with probability at least $1 - \frac{\delta}{8}$ such that, conditioned on $\mathcal{E}_{\mathrm{safe}} \cap \mathcal{E}_{\mathrm{bound}}$
    \begin{equation}
        \|x_{\tauk}\| \le \sqrt{x_{\tauk}^\top \mathsf{dlyap}[A_{\mathrm{cl},\star}] x_{\tauk}} \lesssim \sqrt{\PsiB \Jzero \log(1/\delta)} \Pstarnorm^{3/2}, \quad \forall k \ge \ksafe.
    \end{equation}
\end{lemma}

\begin{lemma}[Lem. 5.4 from \cite{simchowitz20a}]
    Recall $\tauls$ from~\eqref{eq:tauls-def}. There is an event $\Els$, which holds with probability at least $1 - \delta/8$, such that conditioned on $\Els \cap \Esafe \cap \Ebound$,

    \begin{equation}
        \label{OLSrate}
        \Fnorm{\Ahat_k - \Astar}^2 + \Fnorm{\Bhat_k - \Bstar}^2 \lesssim \frac{\dx\du\Pstarnorm^2}{\sigmain^2 \tauk^{1/2}} \log\left(\frac{1}{\delta}\right) + \Pstarnorm^3 \frac{\dx^2}{\tauk} \log\left(\frac{1}{\delta}\right)^2, \quad \forall k : c \tauls \le \tauk \le T,
    \end{equation}
    where $c > 0$ is a universal constant.
\end{lemma}

We now extract the two-rate form of the OLS error bound that motivates the two-scale quantization schedule.

\begin{corollary}[Two-rate OLS error bound]
    \label{cor:ols-two-rate}
    Under the event $\Els \cap \Esafe \cap \Ebound$, for all $k$ such that $c\tauls \le \tauk \le T$,
    \begin{equation}
        \label{eq:ols-two-rate}
        \Fnorm{\Ahat_k - \Astar} + \Fnorm{\Bhat_k - \Bstar} \le \Vslow\,\tauk^{-1/4} + \Vfast\,\tauk^{-1/2},
    \end{equation}
    where
    \begin{align}
        \Vslow &:= \sqrt{C_0}\,\Pstarnorm \sqrt{\frac{\dx\du}{\sigmain^2} \log\!\left(\frac{1}{\delta}\right)}, \label{eq:Vslow-def} \\
        \Vfast &:= \sqrt{C_0}\,\Pstarnorm^{3/2}\,\dx\,\log\!\left(\frac{1}{\delta}\right), \label{eq:Vfast-def}
    \end{align}
    and $C_0$ is a universal constant.
\end{corollary}

\subsection{Data-Driven Computation of the Quantization Schedule}
\label{sec:data-driven}

The constants $\Vslow$ and $\Vfast$ in Corollary~\ref{cor:ols-two-rate} depend on the true $\Pstarnorm$, which is unknown to the algorithm. To set the quantization schedule $\cslow$ and $\cfast$ in a fully implementable way, both the plant and controller use the safe-set geometry available at epoch $\ksafe$ to compute an upper bound.

Applying~\eqref{eq:thm11-P} to the projected center $\tilde\theta_{\ksafe}$~\eqref{eq:safe-proj-overview} gives
\begin{equation}
    \label{eq:Pop-bound}
    \Pstarnorm \le 1.0835\,\opnorm{\Pinf(\Atil_{\ksafe},\Btil_{\ksafe})} =: \Pophat.
\end{equation}
Since $(\Atil_{\ksafe}, \Btil_{\ksafe})$ is the shared decoded model, $\Pophat$ is computable by both sides. Replacing $\Pstarnorm$ with $\Pophat$ in Corollary~\ref{cor:ols-two-rate} yields deterministic upper bounds on the true constants:

\begin{definition}[Empirical proxies]
    \label{def:empirical-proxies}
    \begin{align}
        \Vslowhat &:= \sqrt{C_0}\,\Pophat\sqrt{\frac{\dx\du}{\sigmain^2}\log\!\left(\frac{1}{\delta}\right)}, \label{eq:Vslowhat} \\
        \Vfasthat &:= \sqrt{C_0}\,\Pophat^{3/2}\,\dx\,\log\!\left(\frac{1}{\delta}\right). \label{eq:Vfasthat}
    \end{align}
\end{definition}

\noindent Appendix~\ref{app:lemmas} further gives the comparison:
\begin{equation}
    \label{eq:Pophat-ratio-main}
    \Pstarnorm \le \Pophat \le (1.0835)^2 \Pstarnorm.
\end{equation}
Hence $\Vslowhat \ge \Vslow$ and $\Vfasthat \ge \Vfast$. Note that $\sigmain^2$ is computed and set at epoch $\ksafe$ (see Algorithm~\ref{alg:saferound}), so it does not require bounding.

The implementable two-scale base schedule coefficients are then set at epoch $\ksafe$ as:
\begin{align}
    \cslow &= \frac{1 + 2^{1/4}}{1 - \varrho\,2^{1/4}}\,\Vslowhat, \label{eq:cslow-impl} \\
    \cfast &= \frac{1+\sqrt{2}}{1 - \varrho\,\sqrt{2}}\,\Vfasthat. \label{eq:cfast-impl}
\end{align}
The transient gap between $\ksafe$ and the onset of the asymptotic OLS rate is handled natively by the adaptive multiplier $m_k$, which absorbs any excess innovation (Section~\ref{sec:contraction}).

The leading quantization error is driven by $\cslow^2 \propto \Vslowhat^2 \propto \dx\du$, isolating the $\dx^2$ scaling to $\cfast$ and achieving optimal $\tilde{\mathcal{O}}(\sqrt{\dx\du^2 T})$ regret.

\subsection{Quantization Error Analysis}
\label{sec:contraction}

Standard differential quantizers set a fixed codebook radius $s_k$ and assume the parameter innovation $\|\Delta_k\|_2$ will always land within it. In our setting, the two-scale base schedule~\eqref{eq:two-scale-schedule} is calibrated to the asymptotic OLS error rate (Corollary~\ref{cor:ols-two-rate}), which only holds for epochs $k > k_{\mathrm{ls}}$. During the transient window between $\ksafe$ (when differential tracking begins) and $k_{\mathrm{ls}}$ (when the asymptotic OLS rate bounds take effect), the estimation error may significantly exceed $\sbase$, causing the innovation $\|\Delta_k\|_2$ to overflow the base codebook radius. Enlarging the codebook to cover the worst-case transient error would force the per-epoch bit cost to scale with $\Omega(T)$, exceeding the logarithmic budget. Adaptive quantizer scaling resolves this via the dynamic multiplier $m_k$, which expands the effective quantization radius on-the-fly to absorb any transient innovation excess. The key mathematical insight is that once the asymptotic OLS regime begins, the interaction between the shrinking estimation error and the codebook dynamics forms a \emph{contraction}: $m_k$ collapses back to $\mathcal{O}(1)$ within a constant number of epochs with high probability, achieving overflow immunity during the transient phase without incurring an asymptotic bit-rate penalty.

What remains is to show that: (i) the multiplier $m_k$ contracts to $\mathcal{O}(1)$, ensuring $\mathcal{O}(\log T)$ bit rate, and (ii) the random scale $s_k = m_k \sbase$ does not inflate the regret beyond the ideal base schedule.

\begin{lemma}[Contraction of the Adaptive Multiplier]
    \label{lem:contraction}
    Let $\varrho < 1/\sqrt{2}$ and $\beta := \varrho\sqrt{2} < 1$. On the event $\Els \cap \Esafe \cap \Ebound$, for all epochs $k > k_{\mathrm{ls}}$ (where the asymptotic OLS bound \eqref{eq:ols-two-rate} holds), the multiplier satisfies
    \begin{equation}
        m_k < 1 + \beta(m_{k-1} - 1) + 1.
    \end{equation}
    Consequently, $\limsup_{k\to\infty} m_k \le \frac{2 - \beta}{1-\beta} = \mathcal{O}(1)$.
\end{lemma}

We establish an envelope recurrence $s_k \le \varrho s_{k-1} + U_k$ (where $U_k$ bounds the unquantized estimation errors), allowing us to bound the total quantization-induced regret. The formal statement and proof of this bound are deferred to Appendix~\ref{app:achievability}.

\subsection{Pre-Safe Period Regret}
\label{sec:presafe-regret}

Because our algorithm executes the baseline stabilizing controller $K_0$ before the safety trigger, and this trigger occurs with high probability at a bounded stopping time, the system behavior during this phase is identical to the unquantized setting. Consequently, we can upper-bound the regret accumulated during the pre-safe period using the same method as Lemma~5.5 in \cite{simchowitz20a}.

\begin{lemma}[Pre-Safe Regret]
    \label{lem:presaferegret}
    There is an event $\mathcal{E}_{\mathrm{reg, init}}$ which holds with probability $1 - \frac{\delta}{8}$ such that the duration of the pre-safe exploration phase is bounded by:
    \begin{equation}
        \tau_{\ksafe} \leq 2 \cdot 729^2\,\Pstarnorm^{10}\, d(1 + \opnorm{K_0}^2) \log\frac{\PsiB^2 \Jzero}{\delta},
    \end{equation}
    and the total regret accumulated during this phase is bounded by:
    \begin{equation}
        \label{presaferegretfinal}
        \sum_{t=1}^{\tau_{\ksafe}-1} x_{t}^\top \Rx x_{t} + u_{t}^\top \Ru u_{t} \lesssim d^2 \Pstarnorm^{10} \PsiB^2 \Pzero \left(1 + \opnorm{K_0}^2\right) \log\frac{\PsiB^2 \Jzero}{\delta} \cdot \log\frac{1}{\delta}.
    \end{equation}
\end{lemma}

\noindent\textbf{Remark} (Elimination of the Dimensionality and Delay Penalty).
    Our safety test involves only the OLS confidence bound and then requires the quantization error to be of similar order. This recovers the $\mathcal{O}(\log T)$ stopping time and $\mathcal{O}(\log^2 T)$ pre-safe regret of the unquantized algorithm.

\subsection{Proof of Theorem~\ref{thm:main_ach}}

\begin{proof}[Proof sketch]
    The full proof is given in Appendix~\ref{app:achievability}. The argument proceeds by decomposing the post-safe regret via Lemma~5.2 of \cite{simchowitz20a} into epoch-level cost deviations and boundary state norms. By Lemma~\ref{lem:51}, the cost deviation at each epoch satisfies
    \begin{equation}
        J_k - \Jstar \;\lesssim\; \Pstarnorm^8 \,\Fnorm{\tilde\theta_k - \theta_\star}^2,
    \end{equation}
    which splits into an unquantized OLS error and a quantization reconstruction error $\varrho^2 s_k^2$. The adaptive multiplier $m_k$ makes $s_k$ a random variable, but Lemma~\ref{lem:contraction} ensures $m_k = \mathcal{O}(1)$ asymptotically.

    Applying Lemma~\ref{lem:convolution} bounds the weighted quantization regret by the ideal base schedule: the $\varrho$-contraction in the envelope recurrence $s_k \le \varrho s_{k-1} + U_k$ contributes a factor of $\Crho^2$, so the quantization-induced regret scales as
    \begin{equation}
        \varrho^2 \Crho^2 \cdot \cslow^2 \cdot \sqrt{T} \;+\; \varrho^2 \Crho^2 \cdot \cfast^2 \cdot \log T.
    \end{equation}
    Substituting the definitions of $\cslow$~\eqref{eq:cslow-impl} and $\cfast$~\eqref{eq:cfast-impl} and absorbing the proxy comparison $\Pophat \le (1.0835)^2 \Pstarnorm$ from~\eqref{eq:Pophat-ratio-main} recovers $Q_{\mathrm{slow}}(\varrho)$~\eqref{eq:Qslow} and $Q_{\mathrm{fast}}(\varrho)$~\eqref{eq:Qfast}. Combining with the pre-safe regret (Lemma~\ref{lem:presaferegret}) yields the stated bound.
\end{proof}

\begin{algorithm}
    \caption{Quantized Certainty Equivalent Control with Adaptive Scale}
    \label{alg:quantized_ce_lqr_proj}
    \begin{algorithmic}[1]
        \Require Stabilizing controller $K_0$, confidence $\delta$, codebook $\mathcal{C}$ ($\varrho$-net of unit ball, $\varrho < 1/\sqrt{2}$).
        \State \textbf{Init:} $\mathrm{safe} \gets \textsc{False}$, $\tilde{\theta}_{\ksafe} \gets \bot$, $\mathcal{B}_{\mathrm{safe}} \gets \bot$, $\sigma_{\mathrm{in}} \gets \infty$.
        \State \textbf{Play:} $u_1 \sim \mathcal{N}(0,I)$.

        \For{epoch $k = 2,3,\dots$}
        \State $\tau_k \gets 2^k$
        \vspace{0.4em}

        \State \textbf{\textsc{--- Plant Side ---}}
        \State $(\hat A_k, \hat B_k, \mathbf{\Lambda}_k) \gets \textsc{OLS}(k)$; \quad $\hat{\theta}_k \gets \mathrm{vec}(\hat A_k, \hat B_k)$.

        \If{$\mathrm{safe}$ is \textsc{False}} \Comment{Pre-safe: Decoupled Statistical Trigger}
        \State $\mathrm{Conf}_k \gets 6\,\lambda_{\min}(\mathbf{\Lambda}_k)^{-1}\!\left(d\log 5 + \log\!\left(\frac{4k^2\det(3\mathbf{\Lambda}_k)}{\delta}\right)\right)$
        \If{$\mathbf{\Lambda}_k \succeq I$ \textbf{and} $\sqrt{\mathrm{Conf}_k} \le \frac{1}{9\,\Csafe(\hat A_k, \hat B_k)}$}
            \State $\epsilon_{\mathrm{target}} \gets \frac{1}{9\,\Csafe(\hat A_k, \hat B_k)}$
            \State $\mathrm{safe} \gets \textsc{True}$; \quad $\ksafe \gets k$.
            \State $\tilde\theta_{\ksafe} \gets \textsc{AbsoluteInit}(\hat\theta_k,\, \epsilon_{\mathrm{target}})$ \Comment{Elias Gamma; zero overflow risk}
            \State $\rsafe \gets \frac{1}{3\,\Csafe(\tilde A_{\ksafe},\tilde B_{\ksafe})}$ \Comment{Computed locally from decoded center}
            \State $\sigmain^2, \mathcal{B}_{\mathrm{safe}}, \cslow, \cfast \gets \textsc{SafeRoundInit}(\tilde A_{\ksafe},\tilde B_{\ksafe}, \rsafe, \varrho, \delta)$.
        \EndIf
        \State \textbf{Transmit} $\mathrm{safe}$ flag to Controller. \Comment{No incremental update}

        \Else \Comment{Post-safe: adaptive scale protocol}
        \State $\bar{\theta}_k \gets \Pi_{\mathcal{B}_{\mathrm{safe}}}(\hat{\theta}_k)$ \Comment{Plant-side projection of raw estimate}
        \State $\Delta_k \gets \bar{\theta}_k - \tilde{\theta}_{k-1}$ \Comment{Track against exact shared baseline}
        \State $\sbase \gets \cslow\,\tau_k^{-1/4} + \cfast\,\tau_k^{-1/2}$ \Comment{Ideal base schedule}
        \State $m_k \gets \max\!\left\{1,\;\left\lceil \|\Delta_k\|_2 / \sbase \right\rceil\right\}$ \Comment{Adaptive multiplier}
        \State $s_k \gets m_k \cdot \sbase$
        \State \textbf{Transmit} $m_k$ via Elias Gamma coding ($2\lfloor\log_2 m_k\rfloor + 1$ bits).
        \State $q_k \gets \operatorname*{argmin}_{c\in\mathcal{C}} \|\Delta_k - s_k c\|_2$; \quad $\tilde{\theta}_k \gets \tilde{\theta}_{k-1} + s_k c_{q_k}$.
        \State \textbf{Transmit} $q_k$ to Controller.
        \EndIf

        \vspace{0.4em}
        \State \textbf{\textsc{--- Controller Side ---}}
        \If{$\mathrm{safe}$ is \textsc{False}}
        \State $\tilde K_k \gets K_0$
        \Else
        \If{$k = \ksafe$}
            \State Receive Elias Gamma bits ($E$ and $z_i$); reconstruct absolute $\tilde{\theta}_{\ksafe}$.
            \State $\rsafe \gets \frac{1}{3\,\Csafe(\tilde A_{\ksafe},\tilde B_{\ksafe})}$ \Comment{Computed locally from decoded center}
            \State $\sigmain^2, \mathcal{B}_{\mathrm{safe}}, \cslow, \cfast \gets \textsc{SafeRoundInit}(\tilde A_{\ksafe},\tilde B_{\ksafe}, \rsafe, \varrho, \delta)$.
        \Else
            \State Receive $m_k$; compute $s_k \gets m_k \cdot \sbase$. Receive $q_k$; update $\tilde{\theta}_k \gets \tilde{\theta}_{k-1} + s_k c_{q_k}$.
        \EndIf
        \State $\check{\theta}_k \gets \Pi_{\mathcal{B}_{\mathrm{safe}}}(\tilde{\theta}_k)$ \Comment{Controller-side safe projection}
        \State Compute CE policy $\tilde K_k \gets K_\infty(\check A_k,\check B_k)$. \Comment{Synthesize from safe parameter}
        \EndIf
        \State \textbf{Transmit} $\tilde K_k$ to Plant.

        \vspace{0.4em}
        \State \textbf{\textsc{--- Plant Side (Execution) ---}}
        \State Set exploration variance $\sigma_k^2 \gets \min\{1,\sigma_{\mathrm{in}}^2\tau_k^{-1/2}\}$.
        \For{$t = \tau_k,\dots,\tau_{k+1}-1$}
        \State Apply $u_t = \tilde K_k x_t + \sigma_k g_t$, with $g_t\sim\mathcal{N}(0,I)$.
        \EndFor
        \EndFor
    \end{algorithmic}
\end{algorithm}

\begin{algorithm}
    \caption{OLS($k$)}
    \label{alg:ols}
    \begin{algorithmic}[1]
        \State \textbf{Input:} Examples $\mathbf{x}_{\tau_{k-1}}, \dots, \mathbf{x}_{\tau_k}$, $\mathbf{u}_{\tau_{k-1}}, \dots, \mathbf{u}_{\tau_k-1}$.
        \State \textbf{Return} $(\widehat{A}_k, \widehat{B}_k, \mathbf{\Lambda}_k)$, where
        \[
            \left[ \widehat{A}_k \;\; \widehat{B}_k \right] \leftarrow \left( \sum_{t=\tau_{k-1}}^{\tau_k-1} \mathbf{x}_{t+1} \begin{bmatrix} \mathbf{x}_t \\ \mathbf{u}_t \end{bmatrix}^\top \right) \mathbf{\Lambda}_k^\dagger, \quad \text{and} \quad \mathbf{\Lambda}_k \leftarrow \sum_{t=\tau_{k-1}}^{\tau_k-1} (\mathbf{x}_t, \mathbf{u}_t)(\mathbf{x}_t, \mathbf{u}_t)^\top.
        \]
    \end{algorithmic}
\end{algorithm}

\begin{algorithm}
    \caption{\textsc{AbsoluteInit}$(\hat\theta,\, \epsilon_{\mathrm{target}})$}
    \label{alg:absoluteinit}
    \begin{algorithmic}[1]
        \Require OLS estimate $\hat\theta \in \mathbb{R}^{d_s}$ with $d_s = \dx^2 + \dx\du$, target $\ell_2$ precision $\epsilon_{\mathrm{target}}$.
        \vspace{0.3em}
        \State $\Delta_{\max} \gets 2\epsilon_{\mathrm{target}} / \sqrt{d_s}$ \Comment{Per-coordinate step to meet total $\ell_2$ error}
        \State $E \gets \lceil \log_2(1/\Delta_{\max}) \rceil$; \quad $\Delta \gets 2^{-E}$ \Comment{Dyadic grid step}
        \State \textbf{Transmit} exponent $E$ via Elias Gamma coding.
        \For{$i = 1,\dots,d_s$}
            \State $z_i \gets \lfloor \hat\theta_i / \Delta \rceil$ \Comment{Round to nearest grid point}
            \State \textbf{Transmit} $z_i$ via Signed Elias Gamma coding.
            \State $\tilde\theta_i \gets z_i \cdot \Delta$
        \EndFor
        \State \textbf{Return} $\tilde\theta$. \Comment{$\|\tilde\theta - \hat\theta\|_2 \le \epsilon_{\mathrm{target}}$ by construction}
    \end{algorithmic}
\end{algorithm}

\begin{algorithm}
    \caption{\textsc{SafeRoundInit}$(\widehat{A}, \widehat{B}, \rsafe, \varrho, \delta)$}
    \label{alg:saferound}
    \begin{algorithmic}[1]
        \Require Reconstructed safe pair $(\widehat{A}, \widehat{B})$, absolute safe radius $\rsafe$, codebook radius $\varrho < 1/\sqrt{2}$, failure probability $\delta$.
        \vspace{0.3em}

        \State \textbf{1: Define Safe Set and Empirical Complexity:}
        \State $\mathcal{B}_{\mathrm{safe}} \gets \left\{(A, B): \max\left\{\| A - \widehat{A}\|_{\mathrm{op}}, \| B - \widehat{B}\|_{\mathrm{op}}\right\} \leq \rsafe\right\}$.
        \State $\Pophat \gets 1.0835\,\opnorm{\Pinf(\Ahat, \Bhat)}$.
        \vspace{0.3em}

        \State \textbf{2: Compute Exploration Variance:}
        \vspace{-0.5em}
        \[
            \sigmain^2 := \sqrt{\dx}\,\Pophat^{9/2} \max\left\{1, \opnorm{\Bhat} + \rsafe\right\} \sqrt{\log \frac{\Pophat}{\delta}}
        \]

        \State \textbf{3: Compute Ideal Asymptotic Bounds and Two-Scale Constants:}
        \vspace{-0.5em}
        \begin{align*}
            \Vslowhat &:= \sqrt{C_0}\,\Pophat \sqrt{\frac{\dx\du}{\sigmain^2} \log\!\left(\frac{1}{\delta}\right)}, \quad
            \Vfasthat := \sqrt{C_0}\,\Pophat^{3/2}\,\dx\,\log\!\left(\frac{1}{\delta}\right) \\[0.5em]
            \cslow &:= \frac{1 + 2^{1/4}}{1 - \varrho\,2^{1/4}}\,\Vslowhat \\[0.5em]
            \cfast &:= \frac{1+\sqrt{2}}{1 - \varrho\,\sqrt{2}}\,\Vfasthat
        \end{align*}

        \State \textbf{Return}  $\sigma_{\mathrm{in}}^2$, $\mathcal{B}_{\mathrm{safe}}$, $\cslow$, $\cfast$.

    \end{algorithmic}
\end{algorithm}

\section{Numerical Experiments}
\label{sec:experiments}

We evaluate a practical variant of QCE-LQR on four benchmark systems of increasing complexity to verify that the quantization overhead is consistently small.

\subsection{Benchmark Systems}

\begin{enumerate}[leftmargin=*,itemsep=4pt]
    \item \textbf{Scalar Unstable} $(\dx=1,\;\du=1)$: $x_{t+1} = 1.1\,x_t + u_t + w_t$. Spectral radius $\rho(A) = 1.1$, parameter dimension $\pdim = 2$.

    \item \textbf{Double Integrator} $(\dx=2,\;\du=1)$: $A = \bigl[\begin{smallmatrix} 1 & 1 \\ 0 & 1 \end{smallmatrix}\bigr]$, $B = \bigl[\begin{smallmatrix} 0.5 \\ 1 \end{smallmatrix}\bigr]$. Marginally unstable $(\rho(A) = 1)$, $\pdim = 6$.

    \item \textbf{Inverted Pendulum} $(\dx=2,\;\du=1)$: Linearized pendulum on a cart (mass $0.1$\,kg, length $0.5$\,m), discretized at $\Delta t = 0.05$\,s. Unstable with $\rho(A) = 1.248$, $\pdim = 6$.

    \item \textbf{Boeing 747 Lateral} $(\dx=4,\;\du=2)$: Yaw--sideslip--roll dynamics~\cite{friedland_control} discretized at $\Delta t = 0.5$\,s. Nearly marginally stable $(\rho(A) = 0.996)$, $\pdim = 24$.
\end{enumerate}

All systems use $\Rx = \Ru = I$ and i.i.d.\ process noise $w_t \sim \mathcal{N}(0,I)$.

\subsection{Practical Variant Used in Simulation}

Theoretical QCE-LQR is defined and analyzed in Section~\ref{sec:achievability}. For finite-horizon simulations, we use a practical variant that keeps the same communication architecture and blockwise update structure but replaces several conservative components by implementation-ready alternatives. These modifications are not covered by the main theorem; the goal of this section is to assess practical low-rate performance.

\subsubsection{Data-Driven Bootstrap Trigger}

The theoretical trigger requires $\sqrt{\Conf_k} \le 1/(9\,\Csafe)$ with $\Csafe \lesssim \opnorm{\Pinf}^5$. Table~\ref{tab:trigger} shows this is unreachable in practice: for the Boeing~747 system the required estimation accuracy is $\sim 10^{-11}$, exceeding what OLS can produce in any feasible horizon.

\begin{table}[h]
\centering
\caption{Theoretical safe trigger gap: $\sqrt{\Conf_k}$ vs.\ $2\epsilon_{\mathrm{target}} = 2/(9\,\Csafe)$ at $T=10{,}000$.}
\label{tab:trigger}
\begin{tabular}{lcccc}
\hline
System & $\opnorm{\Pstar}$ & $\Csafe$ & $2\epsilon_{\mathrm{target}}$ & $\sqrt{\Conf_k}$ \\
\hline
Scalar ($a\!=\!1.1$) & 1.77 & $9.5\!\times\!10^{2}$ & $2.3\!\times\!10^{-4}$ & $8.5\!\times\!10^{-2}$ \\
Double Integrator & 3.60 & $3.3\!\times\!10^{4}$ & $6.8\!\times\!10^{-6}$ & $1.2\!\times\!10^{-1}$ \\
Inv.\ Pendulum & 24.0 & $4.3\!\times\!10^{8}$ & $5.2\!\times\!10^{-10}$ & $8.2\!\times\!10^{-2}$ \\
Boeing 747 & 55.9 & $2.9\!\times\!10^{10}$ & $7.5\!\times\!10^{-12}$ & $3.3\!\times\!10^{-1}$ \\
\hline
\end{tabular}
\end{table}

In simulation, we replace this with a \emph{bootstrap stability trigger}. At each epoch boundary, the plant computes the OLS estimate $(\Ahat_k, \Bhat_k)$ and the corresponding certainty-equivalent gain $\Khat_k$. It then draws $N_{\mathrm{mc}} = 50$ Monte Carlo samples from the OLS sampling distribution. Under the linear model $x_{t+1} = [A \mid B]\,z_t + w_t$ with $z_t = [x_t^\top,\, u_t^\top]^\top$ and Gaussian noise $w_t \sim \mathcal{N}(0, \sigma_w^2 I)$, we let each row $\hat\theta_i$ of the OLS estimate be distributed as $\hat\theta_i \sim \mathcal{N}(\theta_{\star,i},\, \sigma_w^2 \mathbf{\Lambda}_k^{-1})$, where $\mathbf{\Lambda}_k = \sum_t z_t z_t^\top$ is the empirical covariance matrix. Each sample $(A^{(j)}, B^{(j)})$ is drawn from this distribution, and we verify that $\rho(A^{(j)} + B^{(j)} \Khat_k) < 0.99$ for every sample. The safe phase is declared only if all sampled systems are stabilized, ensuring that the trigger fires only when the empirical uncertainty is small enough that the learned controller is statistically likely to be safe.

As a secondary safeguard, both controllers employ a \emph{fallback test}: if the state norm exceeds five times the maximum norm observed during pre-safe exploration, the algorithm reverts to the baseline stabilizing controller $K_0$ and resumes exploration.

\subsubsection{Coordinate-Wise Scalar Quantization}

The theoretical analysis uses a $\varrho$-net codebook on the unit sphere $\cS^{\pdim-1}$ to quantize the correction $\Delta_k = \thetabar_k - \thetatil_{k-1}$, transmitting the direction via a codebook index and the scale via Elias Gamma coding. While this achieves optimal rate--distortion scaling, it is difficult to implement in practice. 

We instead quantize each coordinate of $\Delta_k$ independently using a uniform scalar quantizer with step size $1/\sqrt{\tauk}$. This matches the $\mathcal{O}(\tau^{-1/2})$ ``fast'' OLS estimation rate (Corollary~\ref{cor:ols-two-rate}), which in practice governs the convergence of all parameter coordinates---including the $\dx\du$ subspace whose worst-case theoretical rate is $\tau^{-1/4}$. Coordinate~$i$ is encoded as a sign bit plus the Elias Gamma encoding of the integer index $\lceil |\Delta_{k,i}|\sqrt{\tauk}\,\rceil$. Unlike the theorem's two-scale vector quantizer, this simplification is not accompanied by the same worst-case guarantee: using the fast rate instead of the conservative two-scale schedule~\eqref{eq:two-scale-schedule} may inflate the worst-case bit budget from $\mathcal{O}(\pdim\log T)$ to $\mathcal{O}(\pdim\log^2 T)$. For the four benchmarked systems considered here, however, the realized communication cost remains close to $\mathcal{O}(\pdim \log T)$ since OLS converges at the fast $\tau_k^{-\frac{1}{2}}$ rate across all dimensions in these simulations.

\subsubsection{Exploration Decay}

The theoretical exploration rate is $\sigma_k^2 = \min\{1,\,\sigmain^2\,\tauk^{-1/2}\}$, which depends on $\sigmain$ in Point 6 of Lemma~\ref{lem:51}.
For the inverted pendulum algorithm, this evaluates to $\sigmain^2 \approx 1649$. This means $\sigma_k = 1$ (no decay) for all $\tauk \le \sigmain^4 \approx 7.4 \times 10^{12}$, far exceeding our horizon. The excessive exploration noise injects cost through the unstable dynamics, causing linear regret growth. In simulation, we therefore set $\sigmain^2 = 1$ so that exploration decays as $\sigma_k \sim \tauk^{-1/4}$ within the simulation horizon.

\subsection{Results}

We compare two controllers over $T = 10^4$ steps and $50$ independent trials:
\begin{enumerate}[itemsep=2pt]
    \item \textbf{Unquantized CE:} Full-precision certainty-equivalence. The controller receives $\thetahat_k \in \R^{\pdim}$ exactly at each epoch boundary.
    \item \textbf{Practical QCE-LQR:} Quantized certainty-equivalence with the simulation variant above. We report its realized communication empirically.
\end{enumerate}
Table~\ref{tab:results} and Figures~\ref{fig:regret_grid}--\ref{fig:bits_grid} summarize the findings.

\begin{table}[h]
\centering
\renewcommand{\arraystretch}{1.15}
\caption{Practical QCE-LQR vs.\ Unquantized CE ($T=10^4$, $50$ trials, median regret).}
\label{tab:results}
\begin{tabular}{lccrrrc}
\hline
System & $(\dx,\du)$ & $\rho(A)$ & CE Regret & Practical QCE Regret & Overhead & Bits \\
\hline
Scalar Unstable   & $(1,1)$ & $1.10$ & $953$       & $880$       & $-7.6\%$  & $123$   \\
Double Integrator & $(2,1)$ & $1.00$ & $1{,}886$   & $1{,}723$   & $-8.6\%$  & $279$   \\
Inverted Pendulum & $(2,1)$ & $1.25$ & $22{,}761$  & $15{,}856$  & $-30.3\%$ & $232$  \\
Boeing 747        & $(4,2)$ & $1.00$ & $25{,}484$  & $32{,}504$  & $+27.5\%$ & $819$ \\
\hline
\end{tabular}
\end{table}

\begin{figure*}[t]
    \centering
    \begin{minipage}[t]{0.48\textwidth}
        \centering
        \includegraphics[width=\linewidth]{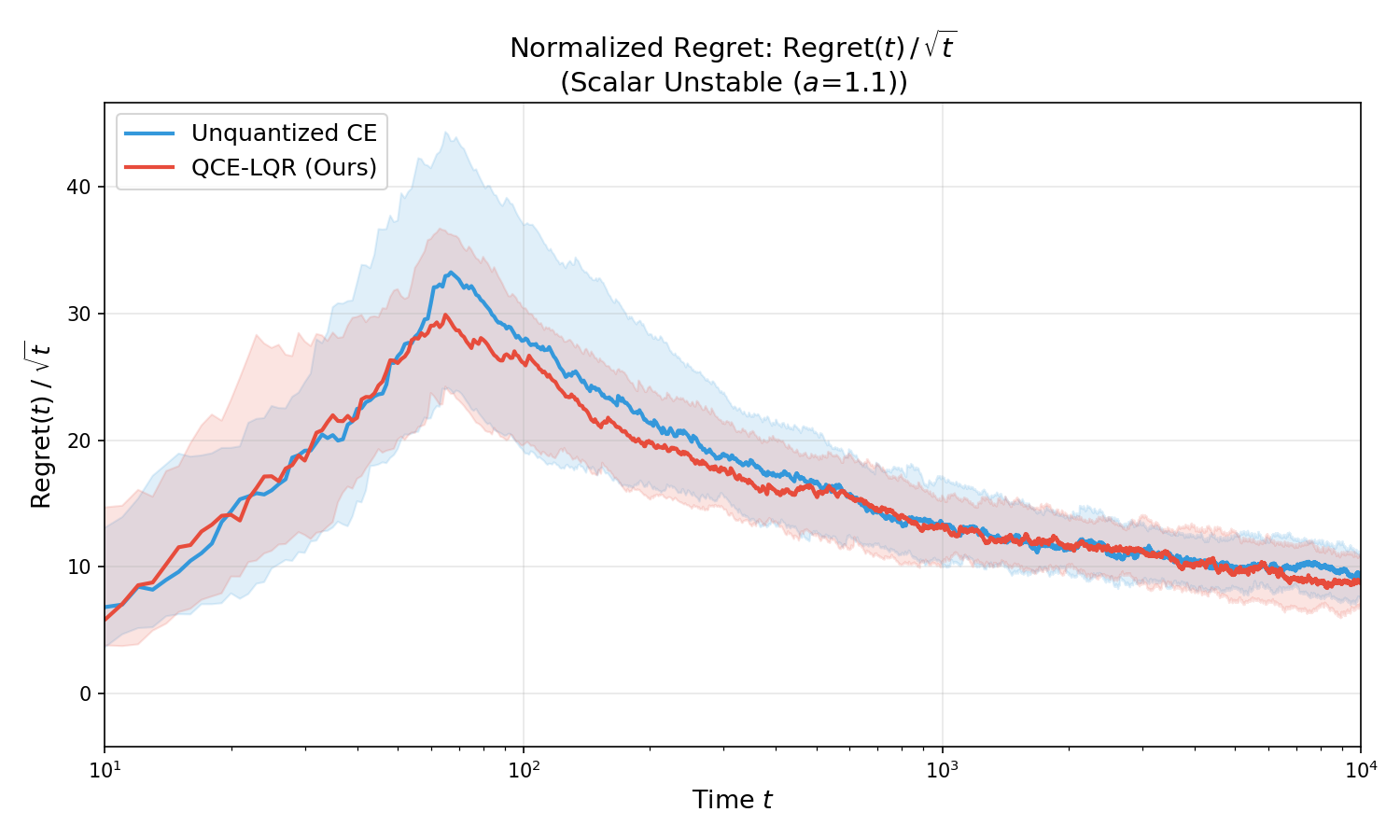}
        {\small (a) Scalar unstable ($a = 1.1$, $123$ bits)}
    \end{minipage}\hfill
    \begin{minipage}[t]{0.48\textwidth}
        \centering
        \includegraphics[width=\linewidth]{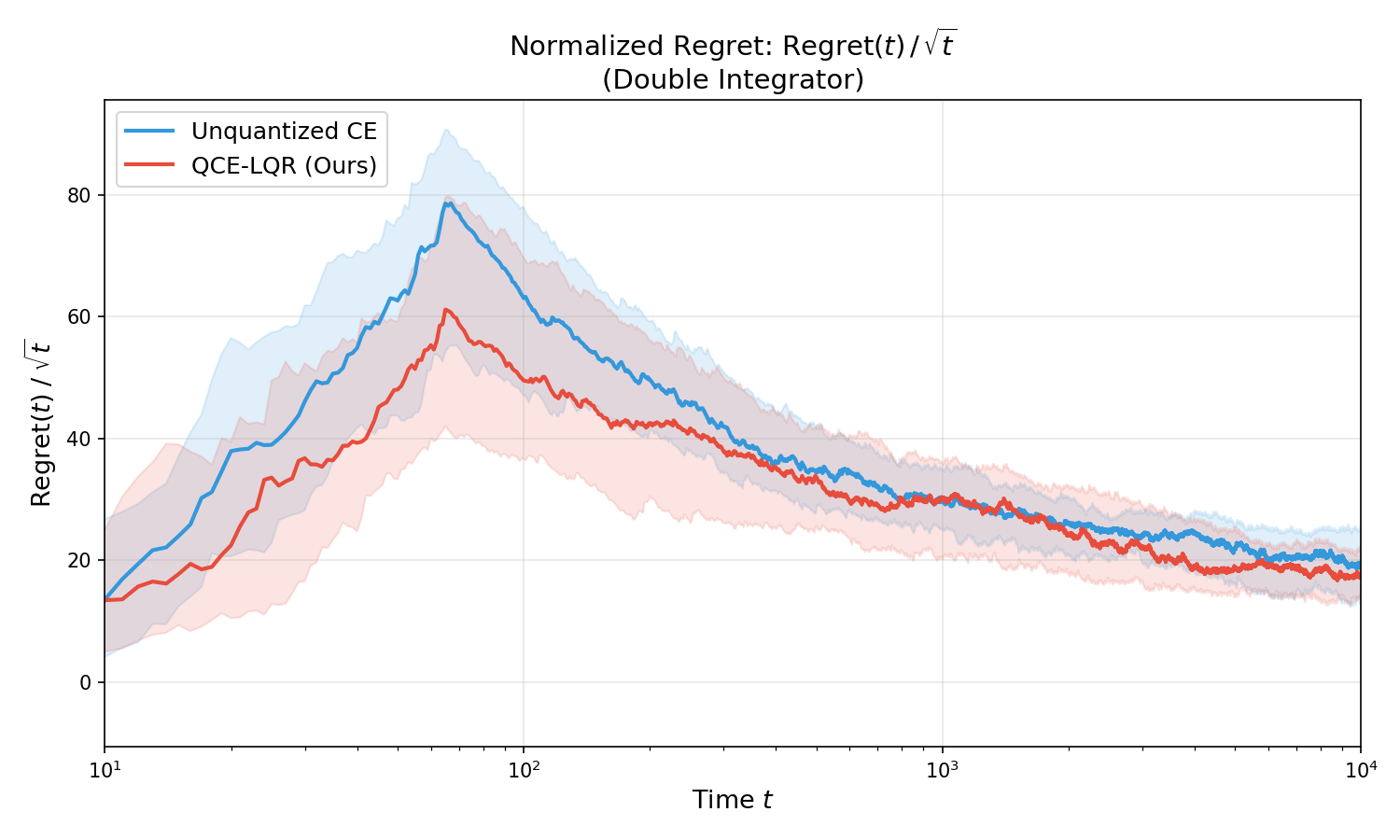}
        {\small (b) Double integrator ($279$ bits)}
    \end{minipage}

    \vspace{6pt}

    \begin{minipage}[t]{0.48\textwidth}
        \centering
        \includegraphics[width=\linewidth]{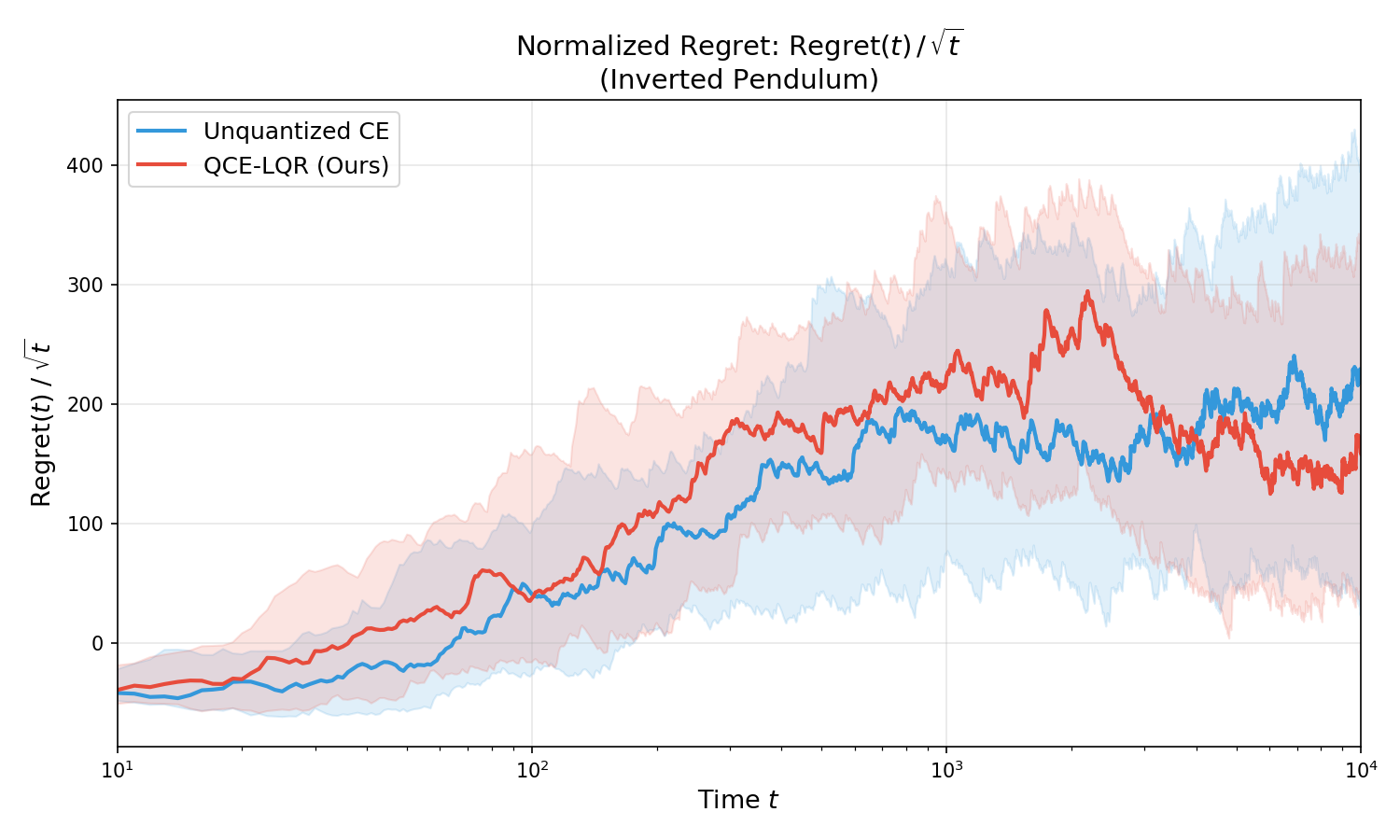}
        {\small (c) Inverted pendulum ($232$ bits)}
    \end{minipage}\hfill
    \begin{minipage}[t]{0.48\textwidth}
        \centering
        \includegraphics[width=\linewidth]{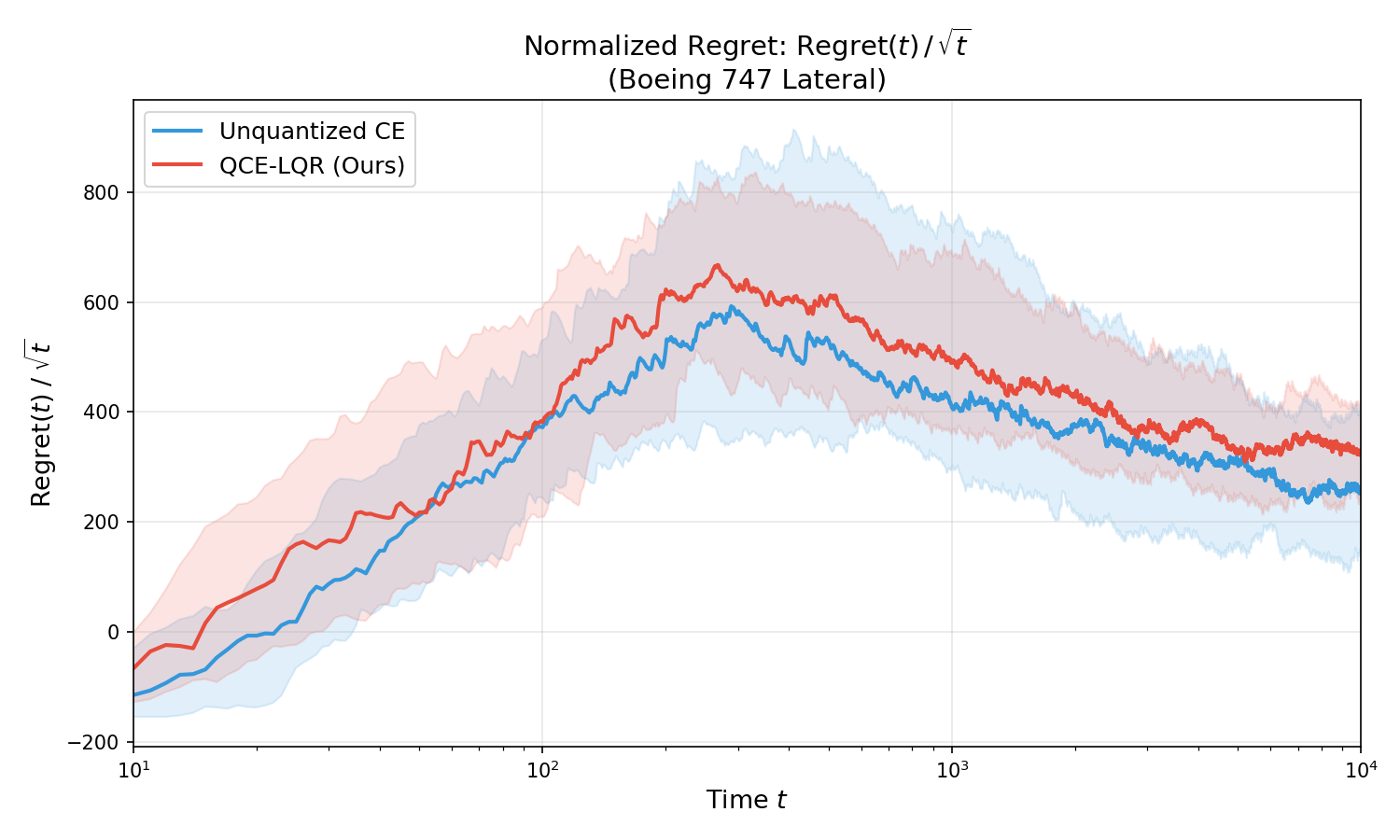}
        {\small (d) Boeing 747 lateral ($819$ bits)}
    \end{minipage}
    \caption{Normalized regret $\Regret(t)/\sqrt{t}$ for Practical QCE-LQR (red) vs.\ unquantized CE (blue), $T=10^4$, $50$ trials. Solid lines show medians; shaded bands show the interquartile range. Both controllers use the causal bootstrap trigger with a runtime fallback shield. In all four systems the two controllers converge to comparable asymptotic values, confirming that quantization overhead is small in these simulations.}
    \label{fig:regret_grid}
\end{figure*}

\begin{figure*}[t]
    \centering
    \begin{minipage}[t]{0.48\textwidth}
        \centering
        \includegraphics[width=\linewidth]{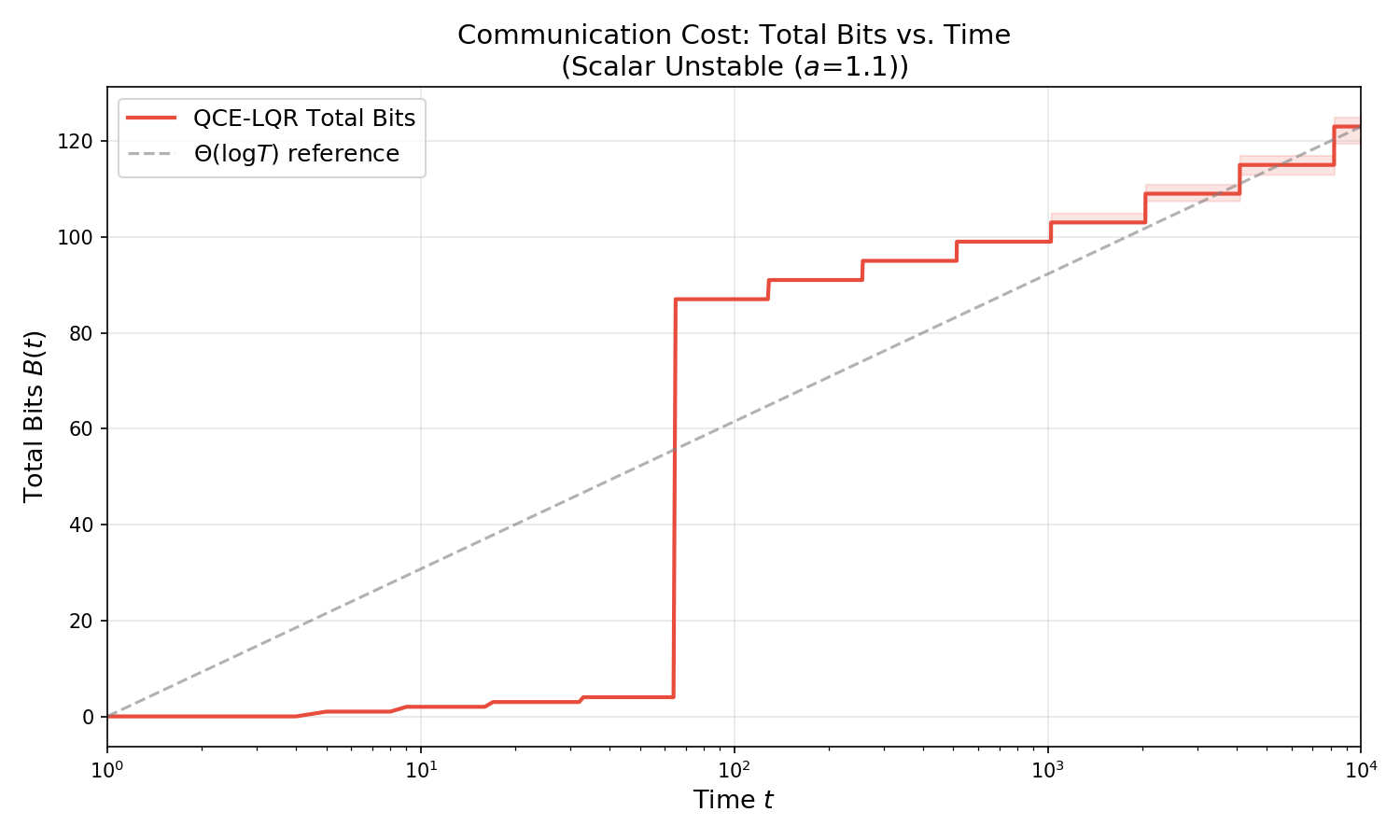}
        {\small (a) Scalar unstable ($123$ bits)}
    \end{minipage}\hfill
    \begin{minipage}[t]{0.48\textwidth}
        \centering
        \includegraphics[width=\linewidth]{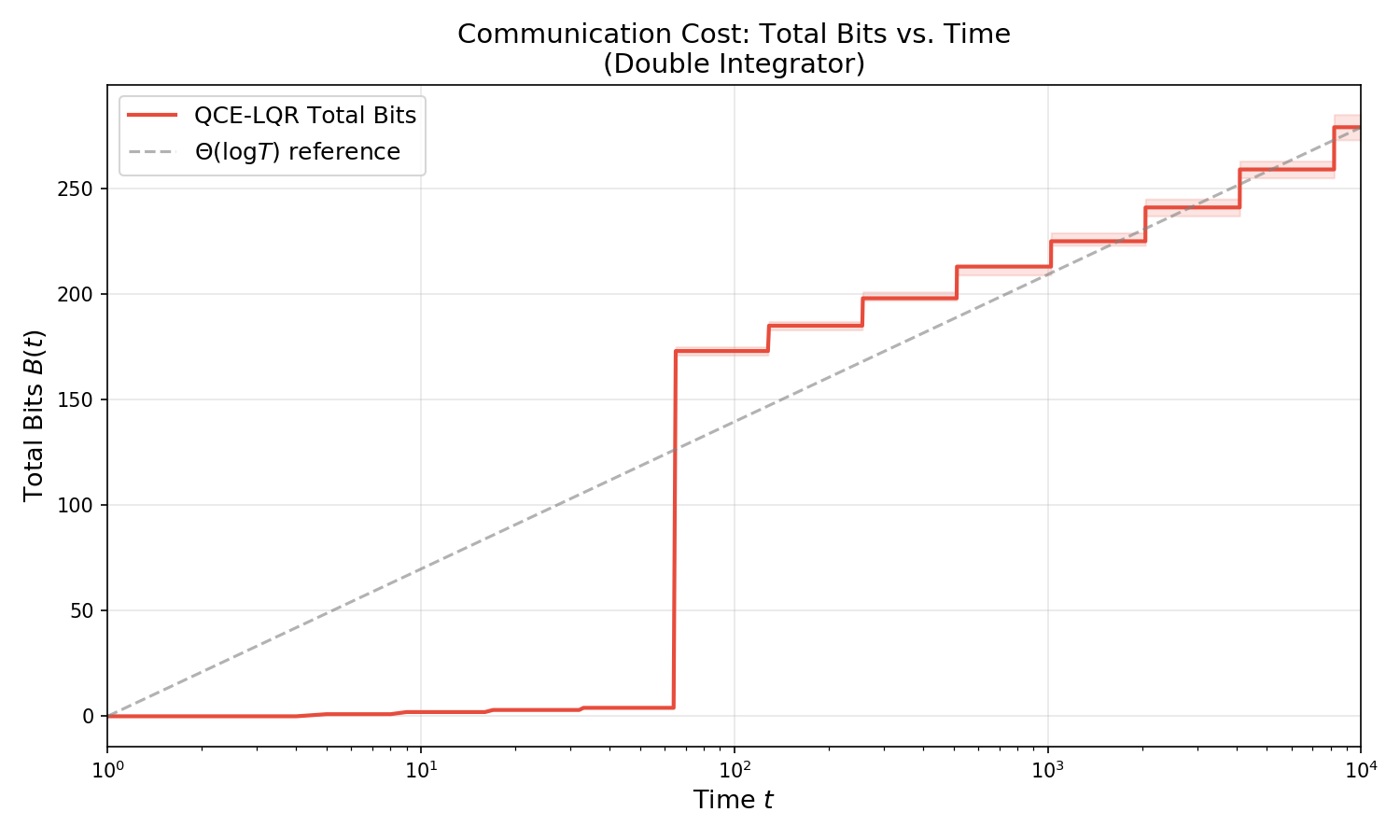}
        {\small (b) Double integrator ($279$ bits)}
    \end{minipage}

    \vspace{6pt}

    \begin{minipage}[t]{0.48\textwidth}
        \centering
        \includegraphics[width=\linewidth]{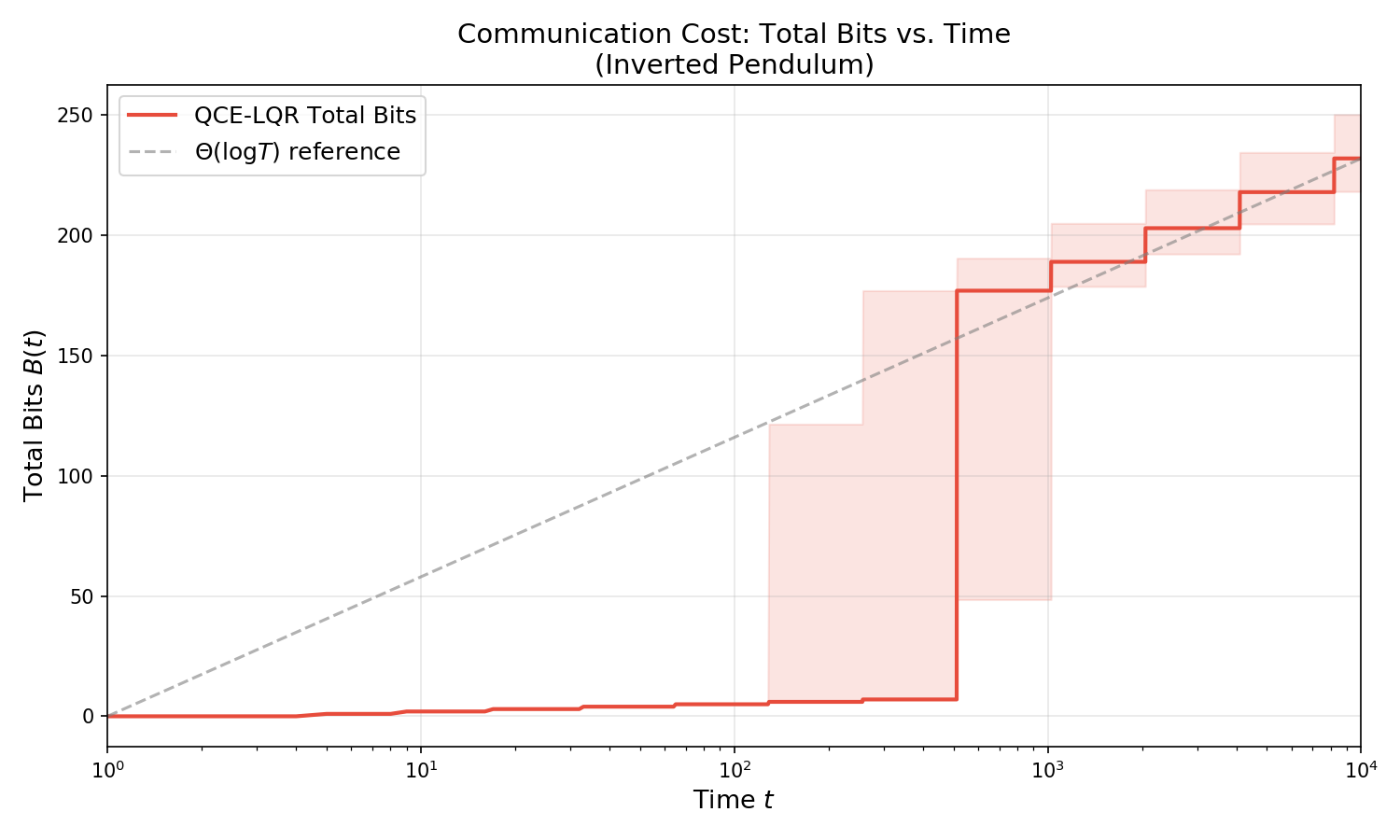}
        {\small (c) Inverted pendulum ($232$ bits)}
    \end{minipage}\hfill
    \begin{minipage}[t]{0.48\textwidth}
        \centering
        \includegraphics[width=\linewidth]{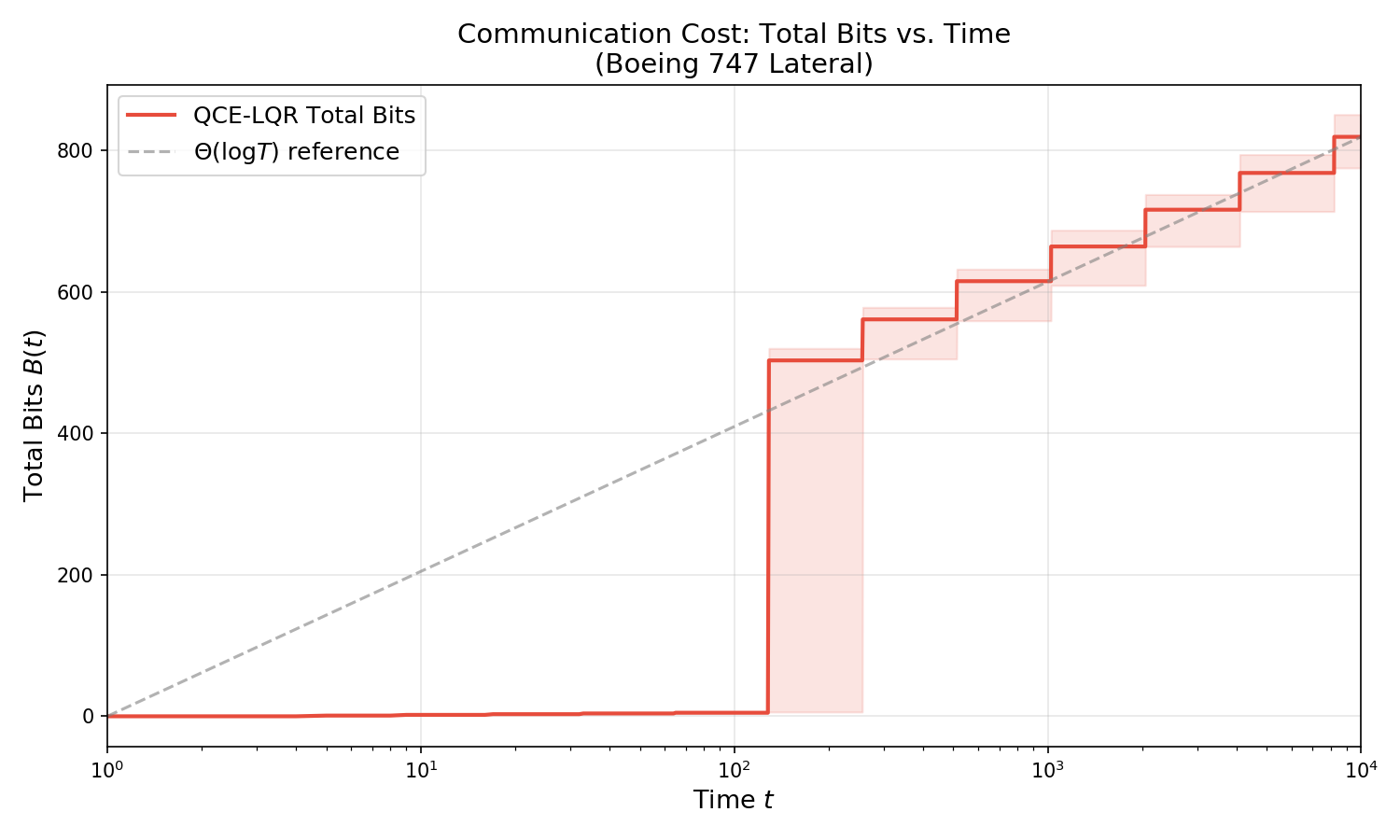}
        {\small (d) Boeing 747 lateral ($819$ bits)}
    \end{minipage}
    \caption{Total bits transmitted $B(t)$ by Practical QCE-LQR over time for each system. Dashed lines show the $\Theta(\log T)$ reference. All systems exhibit the characteristic three-phase structure: pre-safe flags, Elias Gamma initialization jump, and logarithmically growing tracking corrections. Total bits scale with $\pdim$: from $123$ ($\pdim=2$) to $819$ ($\pdim=24$).}
    \label{fig:bits_grid}
\end{figure*}

Across all four benchmarks, Practical QCE-LQR achieves normalized regret comparable to the unquantized baseline (Figure~\ref{fig:regret_grid}). We observe that before the safe flag is declared, the tracking controller incurs linear regret growth since it executes the known stabilizing controller $K_0$. The normalized regret then decays to a constant level as the controller converges to the optimal policy. The differences in median regret (Table~\ref{tab:results}) are within stochastic variability for the lower-dimensional systems. For the Boeing~747 ($\pdim = 24$), the bootstrap trigger fires slightly later in some trials due to the higher statistical uncertainty, resulting in modestly higher regret. The realized bit count scales with $\pdim$: from $123$ bits ($\pdim = 2$) to $819$ bits ($\pdim = 24$), close to the $\mathcal{O}(\pdim \log T)$ trend in these experiments. Figure~\ref{fig:bits_grid} confirms the characteristic three-phase structure across all systems: pre-safe flags, Elias Gamma initialization, and logarithmically growing tracking corrections.

\section{Conclusion}

We have established that $\Theta(\log T)$ bits of communication are both necessary and sufficient to achieve the optimal $\tilde{\mathcal{O}}(\sqrt{T})$ regret in online LQR with unknown dynamics and a rate-limited plant-to-controller link. By transmitting learned models instead of raw states, communication only needs to track the estimator's innovation, which shrinks as learning converges, breaking the $\mathcal{O}(T)$ bandwidth barrier of classical state-quantization approaches.

\subsection*{Summary of Contributions}

\begin{enumerate}
    \item \textbf{Converse (Theorem~\ref{thm:main_converse}).} Any scheme achieving worst-case expected regret $\mathcal{O}(T^\alpha)$ for $\alpha \in [1/2,1)$ must transmit at least $\Omega(\du\dx(1-\alpha)\log T)$ bits by time $T$, even if the true dynamics are known at the system. This establishes a fundamental information-theoretic floor on the communication cost of near-optimal adaptive control.

    \item \textbf{Achievability (Theorem~\ref{thm:main_ach}).} The QCE-LQR algorithm attains $\tilde{\mathcal{O}}(\sqrt{T})$ regret using only $\mathcal{O}((\dx^2 + \dx\du)\log T)$ total bits. The algorithm employs a two-scale adaptive quantization protocol that separately tracks the slow ($\tau_k^{-1/4}$) and fast ($\tau_k^{-1/2}$) decay rates of the OLS estimation error, combined with an absolute initialization via Elias Gamma coding that guarantees zero overflow risk.

    \item \textbf{Explicit Quantization--Regret Tradeoff.} The regret bound in Theorem~\ref{thm:main_ach} exposes the precise cost of quantization through the inflation factors $Q_{\mathrm{slow}}(\varrho)$ and $Q_{\mathrm{fast}}(\varrho)$. As the codebook resolution increases ($\varrho \to 0$), these factors vanish, smoothly recovering the exact unquantized Simchowitz \& Foster baseline.
    \item \textbf{Dimension matching and explicit tradeoff.} The two-scale schedule quarantines the $\dx^2$ scaling into the lower-order $\log T$ term, preserving the optimal $\tilde{\mathcal{O}}(\sqrt{\dx\du^2 T})$ dimension dependence. The quantization penalty is fully captured by $Q_{\mathrm{slow}}(\varrho)$ and $Q_{\mathrm{fast}}(\varrho)$, which vanish as $\varrho \to 0$, making the bandwidth--regret tradeoff explicit.
\end{enumerate}

\subsection*{Next Steps}

\textbf{Closing the dimensional gap.} Our bounds are tight in the horizon $T$ but a gap remains in the dimension-dependent constant multiplying $\log T$. The converse establishes a lower bound of $\Omega(\du\dx\log T)$ bits, counting only the $\du\dx$ degrees of freedom in the gain $K \in \R^{\du\times\dx}$, since the proof technique relies on a mutual information bound between $\Kstar$ and $\hat K$, while the achievability requires $\mathcal{O}((\dx^2 + \dx\du)\log T)$ bits for the full estimate $\hat\theta = \mathrm{vec}(\hat A, \hat B) \in \R^{\dx^2 + \dx\du}$. The additional $\dx^2 \log T$ term arises because the controller needs $\hat A$ to solve the DARE; in our architecture these parameters cannot be eliminated from the uplink. Closing this gap by strengthening the converse to $\Omega((\dx^2 + \dx\du)\log T)$ remains open.

\textbf{Symmetric channel extension.} Our main results assume an asymmetric bi-directional channel (rate-limited uplink, unconstrained downlink). The analysis extends to a fully symmetric rate-limited channel in which the controller must also quantize the policy $\Ktil_k$. The policy innovation $\Delta K_k := \Ktil_k - \Ktil_{k-1}$ inherits the same $\mathcal{O}(\tau_k^{-1/4})$ decay rate as the dynamics estimates, so the controller can quantize $\Delta K_k$ using the identical two-scale schedule and Elias Gamma protocol. The resulting Frobenius-norm error can be bounded using the same techniques. A rate of $\mathcal{O}(\log T)$ downlink bits preserves the optimal $\tilde{\mathcal{O}}(\sqrt{T})$ regret, yielding a total symmetric cost of $\mathcal{O}((\dx^2 + \dx\du)\log T)$ bits. A full proof of this extension is left for future work.

\bibliographystyle{ieeetr}

\begin{thebibliography}{10}

\bibitem{simchowitz20a}
M.~Simchowitz and D.~Foster, ``Naive exploration is optimal for online {LQR},''
  in {\em Proceedings of the 37th International Conference on Machine Learning}
  (H.~D. III and A.~Singh, eds.), vol.~119 of {\em Proceedings of Machine
  Learning Research}, pp.~8937--8948, PMLR, 13--18 Jul 2020.

\bibitem{mania2019}
H.~Mania, S.~Tu, and B.~Recht, ``Certainty equivalence is efficient for linear
  quadratic control,'' in {\em Advances in Neural Information Processing
  Systems}, pp.~10154--10164, Dec 2019.

\bibitem{kargin}
T.~Kargin, S.~Lale, K.~Azizzadenesheli, A.~Anandkumar, and B.~Hassibi,
  ``Thompson sampling for partially observable linear-quadratic control,'' in
  {\em 2023 American Control Conference (ACC)}, pp.~4561--4568, May 2023.

\bibitem{kostina2019}
V.~Kostina and B.~Hassibi, ``Rate-cost tradeoffs in control,'' {\em IEEE
  Transactions on Automatic Control}, vol.~64, pp.~4525--4540, Apr. 2019.

\bibitem{tatikondacontrol}
S.~Tatikonda and S.~Mitter, ``Control under communication constraints,'' {\em
  IEEE Transactions on Automatic Control}, vol.~49, no.~7, pp.~1056--1068, Jul
  2004.

\bibitem{barronisit}
B.~Han, V.~Kostina, B.~Hassibi, and O.~Sabag, ``Coded {K}alman filtering over
  {MIMO} {G}aussian channels with feedback,'' in {\em 2024 IEEE International
  Symposium on Information Theory (ISIT)}, pp.~3261--3266, Jul 2024.

\bibitem{sahaimitter}
A.~Sahai and S.~Mitter, ``The necessity and sufficiency of anytime capacity for
  stabilization of a linear system over a noisy communication link—part i:
  Scalar systems,'' {\em IEEE Transactions on Information Theory}, vol.~52,
  pp.~3369--3395, Jul 2006.

\bibitem{mitra2024}
A.~Mitra, L.~Ye, and V.~Gupta, ``Towards model-free {LQR} control over
  rate-limited channels,'' in {\em Proceedings of the Learning for Dynamics and
  Control Conference (L4DC)}, Proceedings of Machine Learning Research,
  pp.~1253--1265, PMLR, Jul 2024.

\bibitem{fazel18a}
M.~Fazel, R.~Ge, S.~Kakade, and M.~Mesbahi, ``Global convergence of policy
  gradient methods for the linear quadratic regulator,'' in {\em Proceedings of
  the 35th International Conference on Machine Learning} (J.~Dy and A.~Krause,
  eds.), vol.~80 of {\em Proceedings of Machine Learning Research},
  pp.~1467--1476, PMLR, 10--15 Jul 2018.

\bibitem{friedland_control}
B.~Friedland, {\em Control System Design: An Introduction to State-Space
  Methods}.
\newblock Dover Publications, reprint~ed., 2005.

\bibitem{abbasi}
Y.~Abbasi-Yadkori and C.~Szepesv\'ari, ``Regret bounds for the adaptive control
  of linear quadratic systems,'' in {\em Proceedings of the 24th Annual
  Conference on Learning Theory} (S.~M. Kakade and U.~von Luxburg, eds.),
  vol.~19 of {\em Proceedings of Machine Learning Research}, (Budapest,
  Hungary), pp.~1--26, PMLR, 09--11 Jun 2011.

\end{thebibliography}

\newpage
\appendices
\section{Proof of Theorem~\ref{thm:main_converse} (Converse)}
\label{app:converse}

\begin{proof}
As outlined in Section~\ref{sec:converse}, this proof proceeds in four main steps. First, we construct a hard subclass of systems $\Theta_{\mathrm{hard}}$ parametrized by the optimal controller gain $K$. Second, we establish a regret identity demonstrating that sub-linear regret bounds the mean-squared error of estimating $K$. Third, we use the controller's sequence of actions to construct a concrete estimator $\widehat{K}$ for the hidden gain. Finally, we place a uniform prior on the hard class and apply information-theoretic inequalities to show that transmitting enough bits to form this estimator requires $\Omega(\log T)$ bits.

We may assume $T\ge 4$, since finitely many smaller horizons can be absorbed into the constant $C$.

\paragraph{Step 1: Constructing a Hard Subclass.} We begin by defining a localized hypercube of optimal control gains.
Set
\begin{equation}\label{eq:cube-def}
    a:=\frac{r}{\sqrt{\du\dx}},
    \qquad
    \mathcal{C}_r := [-a,a]^{\du\times \dx},
\end{equation}
so that every $K\in\mathcal{C}_r$ satisfies $\|K\|_F\le r$.

\medskip
Choose
\begin{equation}\label{eq:c-choice}
    c > 1 + \frac{r^2\lambda_{\max}(\Ru)}{\lambda_{\min}(\Rx)}
\end{equation}
and set
\begin{equation}\label{eq:P-def-converse}
    P:=c\Rx\succ 0.
\end{equation}
For each $K\in\mathcal{C}_r$, define
\begin{equation}\label{eq:MK-def}
    M_K := (c-1)\Rx - K^\top \Ru K.
\end{equation}
Since $\|K\|_F\le r$,
\begin{equation}\label{eq:KRuK-bound}
    K^\top \Ru K \preceq \lambda_{\max}(\Ru)\|K\|_F^2 I_{\dx}
    \preceq r^2\lambda_{\max}(\Ru) I_{\dx},
\end{equation}
so our choice of $c$ implies $M_K\succ 0$. Now define
\begin{equation}\label{eq:hard-instance-def}
    \Phi_K := P^{-1/2} M_K^{1/2},
    \qquad
    B_K := -P^{-1}\Phi_K^{-\top} K^\top \Ru,
    \qquad
    A_K := \Phi_K - B_KK.
\end{equation}
Then
\begin{align}
    A_K+B_KK &= \Phi_K, \label{eq:ABK-identity} \\
    \Phi_K^\top P\Phi_K &= M_K = (c-1)\Rx - K^\top \Ru K, \label{eq:PhiK-identity} \\
    B_K^\top P\Phi_K &= -\Ru K. \label{eq:BK-cross-term}
\end{align}
Hence
\begin{equation}\label{eq:dare-verification}
    \Rx + K^\top \Ru K + \Phi_K^\top P\Phi_K = P.
\end{equation}
Also, $\Phi_K^\top P\Phi_K \preceq (c-1)\Rx = \frac{c-1}{c}P$,
so $\rho(\Phi_K)<1$. Therefore $(A_K,B_K)$ is stabilizable.

Let
\begin{equation}\label{eq:Theta-hard-def}
    \Theta_{\mathrm{hard}} := \{(A_K,B_K):\ K\in\mathcal{C}_r\}.
\end{equation}
We now show that $K_\infty(A_K,B_K)=K$ for every $K\in\mathcal{C}_r$, so in particular $\Theta_{\mathrm{hard}}\subset \Theta$.

\medskip
\paragraph{Step 2: The Regret Identity.} Next, we relate the expected regret to the mean-squared-error of the controller's sequence of actions. Fix $K\in\mathcal{C}_r$ and consider the corresponding system $(A_K,B_K)$. Define
\begin{equation}\label{eq:V-SK-def}
    V(x):=x^\top P x,
    \qquad
    S_K:=\Ru+B_K^\top P B_K.
\end{equation}
Then $S_K\succ 0$ and $S_K \succeq \Ru \succeq \lambda_{\min}(\Ru) I_{\du}$.
For arbitrary $x\in\R^{\dx}$ and $u\in\R^{\du}$,

\begin{equation}\label{eq:closed-loop-split}
    A_Kx+B_Ku = \Phi_Kx + B_K(u-Kx).
\end{equation}
Therefore
\begin{align}\label{eq:V-expansion}
    V(A_Kx+B_Ku)
    &=
    x^\top \Phi_K^\top P\Phi_K x
    + 2x^\top \Phi_K^\top P B_K (u-Kx)
    + (u-Kx)^\top B_K^\top P B_K (u-Kx).
\end{align}
Consider
\begin{align}\label{eq:complete-square}
    &x^\top \Rx x + u^\top \Ru u + V(A_Kx+B_Ku)-V(x) \nonumber\\
    &\qquad=
    x^\top\bigl(\Rx + \Phi_K^\top P\Phi_K - P\bigr)x
    + u^\top \Ru u
    - 2x^\top K^\top \Ru (u-Kx)
    + (u-Kx)^\top B_K^\top P B_K(u-Kx) \nonumber\\
    &\qquad=
    x^\top\bigl(\Rx+K^\top \Ru K+\Phi_K^\top P\Phi_K-P\bigr)x
    +(u-Kx)^\top S_K(u-Kx) \nonumber\\
    &\qquad=
    (u-Kx)^\top S_K(u-Kx).
\end{align}
The first equality substitutes~\eqref{eq:V-expansion} for $V(A_Kx+B_Ku)$ and applies~\eqref{eq:BK-cross-term} to the cross term. The second expands $u^\top \Ru u = x^\top K^\top \Ru Kx + 2x^\top K^\top \Ru(u-Kx) + (u-Kx)^\top \Ru(u-Kx)$: the $\pm 2x^\top K^\top \Ru(u-Kx)$ terms cancel and $\Ru + B_K^\top P B_K = S_K$~\eqref{eq:V-SK-def}. The third applies~\eqref{eq:dare-verification}.
If $w\sim\mathcal N(0,\sigma_w^2 I_{\dx})$, then
$\E[V(A_Kx+B_Ku+w)] = V(A_Kx+B_Ku) + \sigma_w^2\Tr(P)$,
so with $J_P := \sigma_w^2\Tr(P)$ we have the one-step identity
\begin{equation}
\label{eq:bellman-converse}
    x^\top \Rx x + u^\top \Ru u
    + \E\!\left[V(A_Kx+B_Ku+w)\right]
    - V(x) - J_P
    =
    (u-Kx)^\top S_K(u-Kx).
\end{equation}
Since $S_K\succ 0$, the unique minimizer is $u=Kx$. Hence
$K_\infty(A_K,B_K)=K$ and $J_\infty(A_K,B_K)=J_P$,
confirming $\Theta_{\mathrm{hard}}\subset \Theta$.

\medskip
Fix $K\in\mathcal{C}_r$ and run any scheme achieving the regret bound~\eqref{eq:regret-assumption} on $(A_K,B_K)$. Write
\begin{equation}\label{eq:regret-eps-def}
    \mathrm{Regret}_T^K := \sum_{t=1}^T c_t - TJ_P,
    \qquad
    \epsilon_t^K := \E\!\left[(u_t-Kx_t)^\top S_K(u_t-Kx_t)\right].
\end{equation}
Summing \eqref{eq:bellman-converse} from $t=1$ to $T$ with $x_1=0$ gives
\begin{equation}
\label{eq:regret-identity-converse}
    \sum_{t=1}^T \epsilon_t^K
    =
    \E[\mathrm{Regret}_T^K] + \E[V(x_{T+1})].
\end{equation}
Since $P=c\Rx$ by~\eqref{eq:P-def-converse}, we have $x^\top \Rx x = \frac{1}{c} V(x)$ exactly. Identity \eqref{eq:bellman-converse} therefore implies
\begin{align}\label{eq:lyapunov-recursion}
    \E[V(x_{t+1})]
    &\le
    \left(\frac{c-1}{c}\right)\E[V(x_t)] + J_P + \epsilon_t^K.
\end{align}
Iterating from $V(x_1)=0$ yields
$\E[V(x_{T+1})]
\le
c J_P + \sum_{t=1}^T \left(\frac{c-1}{c}\right)^{\!T-t}\epsilon_t^K$,
and combining with \eqref{eq:regret-identity-converse} gives
\begin{equation}
\label{eq:weighted-bound-converse}
    \sum_{t=1}^T \left(1-\left(\frac{c-1}{c}\right)^{\!T-t}\right)\epsilon_t^K
    \le
    \E[\mathrm{Regret}_T^K] + c J_P.
\end{equation}
Let $m := \lfloor T/2\rfloor$.
For every $t\in\{2,\dots,m\}$ we have $T-t\ge 1$, hence $1-\left(\frac{c-1}{c}\right)^{T-t} \ge 1-\frac{c-1}{c} = \frac{1}{c}$.
Therefore \eqref{eq:weighted-bound-converse} implies
\begin{equation}
\label{eq:first-half-eps-converse}
    \frac{1}{c}\sum_{t=2}^m \epsilon_t^K
    \le
    \E[\mathrm{Regret}_T^K] + c J_P.
\end{equation}
Since $\Theta_{\mathrm{hard}}\subset \Theta$, the assumed regret bound gives $\E[\mathrm{Regret}_T^K] \le C_1 T^\alpha$ for all $K\in\mathcal{C}_r$.

\medskip
Fix $t\ge 2$ and let $\mathcal F_{t-1}$ denote the sigma-field generated by $(x_1,u_1,\ldots,x_{t-1},u_{t-1})$. Let $\mathcal G_{t-1}$ be the enlargement of $\mathcal F_{t-1}$ by all encoder/controller randomization used before $u_t$ is chosen. By the timing convention in Section~\ref{sec:setting}, the transcript available when selecting $K_t$ is generated from the past state-action history, so $K_t$ is $\mathcal G_{t-1}$-measurable. Since
$x_t = A_Kx_{t-1} + B_Ku_{t-1} + w_{t-1}$
with $w_{t-1}\sim\mathcal N(0,\sigma_w^2 I_{\dx})$ independent of $\mathcal G_{t-1}$, we have
$\E[x_t x_t^\top \mid \mathcal G_{t-1}] \succeq \sigma_w^2 I_{\dx}$.
Writing $M_t := (K_t-K)^\top S_K (K_t-K)\succeq 0$ and using $u_t=K_tx_t$,
\begin{align}\label{eq:eps-lower-bound}
    \epsilon_t^K
    &=
    \E\!\left[x_t^\top M_t x_t\right]
    =
    \E\!\left[\Tr(M_t x_t x_t^\top)\right] \nonumber\\
    &=
    \E\!\left[\Tr\!\left(M_t\,\E[x_t x_t^\top \mid \mathcal G_{t-1}]\right)\right] \nonumber\\
    &\ge
    \sigma_w^2\,\E\!\left[\Tr(M_t)\right] \nonumber\\
    &\ge
    \sigma_w^2\lambda_{\min}(\Ru)\,\E\|K_t-K\|_F^2.
\end{align}
Set $c_0 := \sigma_w^2\lambda_{\min}(\Ru)>0$.
Combining \eqref{eq:first-half-eps-converse} with the regret hypothesis yields
\begin{equation}
\label{eq:estimator-mse-converse}
    \frac{1}{m-1}\sum_{t=2}^m \E\|K_t-K\|_F^2
    \le
    C_{\mathrm{est}} T^{\alpha-1},
\end{equation}
where
$C_{\mathrm{est}} := \frac{5c}{c_0}\left(C_1 + c J_P\right)$,
using $T\ge 4 \Rightarrow m-1\ge T/5$.

\medskip
\paragraph{Step 3: Constructing the Gain Estimator.} Achieving sub-linear regret implies that the chosen controls $K_t$ are ultimately very close to the true optimal gain $K$. Let $Z$ be uniform on $\{2,\dots,m\}$ and independent of everything else, and define $\widehat K := K_Z$.
Averaging \eqref{eq:estimator-mse-converse} over $Z$ gives
\begin{equation}
\label{eq:Khat-error}
    \E\|\widehat K-K\|_F^2 \le C_{\mathrm{est}} T^{\alpha-1}.
\end{equation}

\medskip
\paragraph{Step 4: Information-Theoretic Lower Bound.} Finally, we apply the data-processing inequality to bound the mutual information between the communication transcript and the true system gain. Place the prior $K\sim \mathrm{Unif}(\mathcal{C}_r)$ on the hard family $\Theta_{\mathrm{hard}}$. Let $\Pi_T$ denote the full uplink transcript received by the controller by time $T$.
Since $(Z,\widehat K)$ is obtained from $\Pi_T$, the Markov chain
$K \to \Pi_T \to (Z,\widehat K)$
holds. By the data-processing inequality,
\begin{equation}
\label{eq:info-upper-converse}
    \I(K;Z,\widehat K)
    \le
    \I(K;\Pi_T)
    \le
    H(\Pi_T)
    \le
    B(T)+1.
\end{equation}
Since $K$ is uniform on $[-a,a]^{\du\dx}$,
$h(K) = \du\dx\log_2(2a) = \du\dx\log_2\!\left(\frac{2r}{\sqrt{\du\dx}}\right)$.
Also,
$h(K\mid Z,\widehat K) \le h(K-\widehat K)$.
By the Gaussian maximum-entropy bound under a second-moment constraint and \eqref{eq:Khat-error},
\begin{equation}\label{eq:gauss-maxent}
    h(K-\widehat K)
    \le
    \frac{\du\dx}{2}\log_2\!\left(\frac{2\pi e}{\du\dx}\,C_{\mathrm{est}}T^{\alpha-1}\right).
\end{equation}
Therefore
\begin{align}\label{eq:final-bound}
    B(T)+1
    &\ge
    h(K)-h(K\mid Z,\widehat K) \nonumber\\
    &\ge
    \du\dx\log_2\!\left(\frac{2r}{\sqrt{\du\dx}}\right)
    - \frac{\du\dx}{2}\log_2\!\left(\frac{2\pi e}{\du\dx}\,C_{\mathrm{est}}T^{\alpha-1}\right) \nonumber\\
    &=
    \frac{\du\dx}{2}(1-\alpha)\log_2 T
    + \du\dx\log_2\!\left(\frac{2r}{\sqrt{\du\dx}}\right)
    - \frac{\du\dx}{2}\log_2\!\left(\frac{2\pi e}{\du\dx}\,C_{\mathrm{est}}\right).
\end{align}
This gives $B(T) \ge \frac{\du\dx}{2}(1-\alpha)\log_2 T - C$
for the finite constant
$C := 1 + \frac{\du\dx}{2}\log_2\!\left(\frac{2\pi e}{\du\dx}\,C_{\mathrm{est}}\right) - \du\dx\log_2\!\left(\frac{2r}{\sqrt{\du\dx}}\right)$.
\end{proof}

\section{Proof of Supporting Lemmas}
\label{app:lemmas}

\begin{proof}[Proof of Lemma \ref{lem:51}]

    By Algorithm \ref{alg:quantized_ce_lqr_proj}, at epoch $k_\safe$, the Decoupled Statistical Trigger passes:
    \begin{equation}
        \label{safecond}
        \sqrt{\mathrm{Conf}_{k_\safe}} \le \epsilon_{\mathrm{target}} = \frac{1}{9\,\Csafe(\hat A_{k_\safe},\hat B_{k_\safe})}.
    \end{equation}
    The plant then quantizes $\hat\theta_{k_\safe}$ at precision $\epsilon_{\mathrm{target}}$~\eqref{eq:eps-target-def}, yielding a shared reconstruction $\tilde\theta_{k_\safe}$ satisfying
    \begin{equation}\label{eq:quant-guarantee}
        \|\tilde\theta_{k_\safe} - \hat\theta_{k_\safe}\|_2 \ \le \epsilon_{\mathrm{target}}.
    \end{equation}
    Recall from~\eqref{eq:rsafe-def} that both sides set $\rsafe = 1/(3\,\Csafe(\tilde A_{\ksafe}, \tilde B_{\ksafe}))$.

    Note that the estimated safe constant $\Csafe(\hat A_{\ksafe},\hat B_{\ksafe})$ remains finite at the trigger, meaning the center pair $(\hat A_{\ksafe}, \hat B_{\ksafe})$ is stabilizable. Since $\epsilon_{\mathrm{target}} < \frac{1}{3\,\Csafe(\hat A_{\ksafe},\hat B_{\ksafe})}$, we may apply Lemma~\ref{lem:11} with center $(A_0, B_0) = (\hat A_{\ksafe}, \hat B_{\ksafe})$ to the decoded pair $(\tilde A_{\ksafe}, \tilde B_{\ksafe})$, concluding that it is stabilizable with
    \begin{equation}\label{eq:Csafe-bootstrap}
        \Csafe(\tilde A_{\ksafe}, \tilde B_{\ksafe}) \le 1.5\,\Csafe(\hat A_{\ksafe},\hat B_{\ksafe}).
    \end{equation}
    Applying Lemma~\ref{lem:11} with the same center $(\hat A_{\ksafe}, \hat B_{\ksafe})$ to the true system $(A,B)$ on the event $\Esafe$ similarly yields
    \begin{equation}\label{eq:Csafe-true}
        \Csafe(A,B) \le 1.5\,\Csafe(\hat A_{\ksafe},\hat B_{\ksafe}).
    \end{equation}
    Combining~\eqref{safecond} and~\eqref{eq:quant-guarantee} with~\eqref{eq:Csafe-bootstrap}, the total OLS-plus-quantization error satisfies
    \begin{equation}\label{eq:rsafe-bound}
        \sqrt{\Conf_{\ksafe}} + \epsilon_{\mathrm{target}} \le \frac{2}{9\,\Csafe(\hat A_{\ksafe},\hat B_{\ksafe})} \le \frac{1}{3\,\Csafe(\tilde A_{\ksafe}, \tilde B_{\ksafe})} = \rsafe.
    \end{equation}
    We now verify that, on the event $\Esafe$~\eqref{eq:Esafe}, the true system lies in $\mathcal{B}_{\safe}$~\eqref{eq:Bsafe-def}. By the triangle inequality,
    \begin{subequations}
        \label{eq:safety_derivation}
        \begin{align}
            \max \left\{ \| A - \widetilde{A}_{k_{\text{safe}}} \|_{\mathrm{op}} , \| B - \widetilde{B}_{k_{\text{safe}}} \|_{\mathrm{op}} \right\}
             & \leq \left\| \left[ \widetilde{A}_{k_{\text{safe}}} - A \mid \widetilde{B}_{k_{\text{safe}}} - B \right] \right\|_{\mathrm{op}} \label{eq:block_norm}                              \\
             & \leq \left\| \left[ \widehat{A}_{k_{\text{safe}}} - \widetilde{A}_{k_{\text{safe}}} \mid \widehat{B}_{k_{\text{safe}}} - \widetilde{B}_{k_{\text{safe}}} \right] \right\|_{\mathrm{op}} \notag \\
             & \quad + \left\| \left[ \widehat{A}_{k_{\text{safe}}} - A \mid \widehat{B}_{k_{\text{safe}}} - B \right] \right\|_{\mathrm{op}} \label{eq:triangle_ineq}                            \\
             & \leq \epsilon_{\mathrm{target}} + \sqrt{\text{Conf}_{k_\safe}} \label{eq:conf_bound}                                                                                                             \\
             & \leq \rsafe = \frac{1}{3\,\Csafe(\tilde A_{\ksafe}, \tilde B_{\ksafe})}, \label{eq:safe_radius}
        \end{align}
    \end{subequations}
    where \eqref{eq:triangle_ineq} uses the triangle inequality, \eqref{eq:conf_bound} uses~\eqref{eq:quant-guarantee} and \eqref{eq:Esafe}, and \eqref{eq:safe_radius} follows from~\eqref{eq:rsafe-bound}. Therefore both $(A,B)$ and every $(\tilde A, \tilde B) \in \mathcal{B}_\safe$~\eqref{eq:Bsafe-def} lie within a ball of radius $\frac{1}{3\,\Csafe(\tilde A_{\ksafe}, \tilde B_{\ksafe})}$ around the decoded center $(\tilde A_{\ksafe}, \tilde B_{\ksafe})$, which is the hypothesis of Lemma~\ref{lem:11}.

    Moreover, combining~\eqref{safecond}, \eqref{eq:quant-guarantee}, and~\eqref{eq:Csafe-true} gives
    \begin{equation}\label{eq:decoded-near-true}
        \max \left\{ \| A - \widetilde{A}_{k_{\text{safe}}} \|_{\mathrm{op}} , \| B - \widetilde{B}_{k_{\text{safe}}} \|_{\mathrm{op}} \right\}
        \le \frac{2}{9\,\Csafe(\hat A_{\ksafe},\hat B_{\ksafe})}
        \le \frac{1}{3\,\Csafe(A,B)}.
    \end{equation}
    Applying Lemma~\ref{lem:11} once more with center $(A,B)$ and system $(\tilde A_{\ksafe}, \tilde B_{\ksafe})$ yields
    \begin{equation}\label{eq:Pophat-ratio}
        \Pophat = 1.0835\,\opnorm{\Pinf(\Atil_{\ksafe},\Btil_{\ksafe})} \le (1.0835)^2\,\Pstarnorm.
    \end{equation}

    Further,  any $(\tilde A, \tilde B) \in \mathcal{B}_\safe$ satisfies
    \begin{equation}
        \max \left \{ \| \tilde{A} - A \|_{\mathrm{op}}, \| \tilde{B} - B \|_{\mathrm{op}} \right\} \leq \frac{1}{C_\safe(A, B)},
    \end{equation}
    where $\mathcal{B}_{\mathrm{safe}}$ is defined in~\eqref{eq:Bsafe-def}. In particular, in future epochs $k \geq k_\safe$, the quantized system estimates are projected on $\mathcal{B}_\safe$ (Line 25 Algorithm \ref{alg:quantized_ce_lqr_proj}), so it holds for those estimates.

    Then, the first three points of the lemma follow from Theorem 5 in \cite{simchowitz20a}. In particular, for point $1)$,
    \begin{equation}
        J_k - \Jstar \lesssim  \Pstarnorm^8 \left(\|\Pi_{\mathcal{B}_\safe} (\tilde{A}) - A\|_{\text{F}}^2 + \|\Pi_{\mathcal{B}_\safe}(\tilde{B}) - B\|_{\text{F}}^2\right) \leq \Pstarnorm^8 \left(\|\tilde{A} - A\|_{\text{F}}^2 + \|\tilde{B} - B\|_{\text{F}}^2\right),
    \end{equation}
    where the second equality holds since $(A, B) \in \mathcal{B}_\safe$, and the projection is non-expansive.

    The next two points follow from Theorem 8 in \cite{simchowitz20a}. The final point results from the proof of Lemma 5.1 in \cite{simchowitz20a}, recalling that we use the same scaling:
    \begin{equation}\label{eq:sigmain-def}
        \sigmain^2 := \sqrt{\dx} \opnorm{\Pinf(\Ahat, \Bhat)}^{9/2} \max\left\{1, \opnorm{\Bhat} + r_\safe\right\} \sqrt{\log \frac{\opnorm{\Pinf(\Ahat, \Bhat)}}{\delta}}.
    \end{equation}

\end{proof}

\section{Proof of Theorem~\ref{thm:main_ach}: Communication Bound}
\label{app:comm-bound}

\begin{proof}
Let $K_T := \lceil \log_2 T \rceil + 1$ be the number of doubling epochs intersecting $[1,T]$. Since the algorithm communicates only at epoch boundaries, every uplink message is counted once per epoch. We account for four message types.

\medskip\noindent\textbf{Codeword indices.}
The codebook $\mathcal{C}$ is a $\varrho$-net of the unit ball in $\mathbb{R}^{d_s}$. The standard volumetric covering bound gives $|\mathcal{C}| \le (1 + 2/\varrho)^{d_s}$. Each index $q_k$ costs at most
\[
    \log_2 |\mathcal{C}| \le d_s \log_2\!\left(1 + \frac{2}{\varrho}\right)
\]
bits, and over at most $K_T$ post-safe epochs the total is
\begin{equation}\label{eq:Bq}
    B_q(T) \le d_s \log_2\!\left(1 + \frac{2}{\varrho}\right) K_T.
\end{equation}

\medskip\noindent\textbf{Adaptive multipliers.}
Algorithm~\ref{alg:quantized_ce_lqr_proj} sends $m_k$ via Elias Gamma coding at cost $2\lfloor \log_2 m_k \rfloor + 1$ bits. By Lemma~\ref{lem:contraction}, for $k > k_{\mathrm{ls}}$,
\[
    m_k < 2 + \beta(m_{k-1} - 1), \qquad \beta := \varrho\sqrt{2} < 1,
\]
so $\limsup_k m_k \le m_\infty := (2 - \beta)/(1 - \beta)$. Define $M_\varrho := \lceil m_\infty \rceil$. Since $m_k - m_\infty < \beta(m_{k-1} - m_\infty)$, after at most $h = \mathcal{O}(\log_+(m_{k_{\mathrm{ls}}} - m_\infty))$ additional epochs, $m_k \le M_\varrho$ for all subsequent $k$. Therefore all but finitely many multiplier messages cost at most
\[
    b_\varrho := 2\lfloor \log_2 M_\varrho \rfloor + 1
\]
bits each, giving
\begin{equation}\label{eq:Bm}
    B_m(T) \le b_\varrho K_T + \mathcal{O}(\log^2 m_{k_{\mathrm{ls}}}).
\end{equation}
Since $m_k$ decays geometrically over $h = \mathcal{O}(\log_+ m_{k_{\mathrm{ls}}})$ epochs, the sum of the Elias encoding costs $\mathcal{O}(\log m_k)$ forms an arithmetic progression bounded by $\mathcal{O}(\log^2 m_{k_{\mathrm{ls}}})$. Because $m_{k_{\mathrm{ls}}}$ scales polynomially with $\tauls$, this, along with the pre-contraction expansion phase, adds a finite, one-time transient overhead of $\mathcal{O}(\log^2 \tauls)$, which is retained explicitly in~\eqref{eq:total-bits-proof}. Under the theorem choice $\delta = 1/T$, we have $\tauls = \mathcal{O}(\log T)$, so this overhead is at most $\mathcal{O}(\log^2\!\log T)$.

\medskip\noindent\textbf{Pre-safe flags.}
Algorithm~\ref{alg:quantized_ce_lqr_proj} sends a 1-bit safe flag during the pre-safe phase. Thus $B_{\mathrm{flag}} \le \ksafe$. Since $\tau_k = 2^k$, Lemma~\ref{lem:presaferegret} gives
\begin{equation}\label{eq:Bflag}
    \ksafe \le 1 + \log_2 \tau_{\ksafe} = \mathcal{O}\!\left(\log\!\left(d(1+\opnorm{K_0}^2)\Pstarnorm^{10}\log\frac{\PsiB^2 \Jzero}{\delta}\right)\right),
\end{equation}
which is a one-time overhead.

\medskip\noindent\textbf{Absolute initialization.}
\textsc{AbsoluteInit} sends the exponent $E$ and the signed integers $z_1,\dots,z_{d_s}$ via Elias Gamma coding. Both sides compute $r_\safe$ from the shared decoded center~\eqref{eq:rsafe-def}. Since
\[
    E = \left\lceil \log_2 \frac{\sqrt{d_s}}{2\epsilon_{\mathrm{target}}} \right\rceil, \qquad \epsilon_{\mathrm{target}} = \frac{1}{9\,\Csafe(\hat A_{\ksafe}, \hat B_{\ksafe})},
\]
and by Lemma \ref{lem:11}, $\Csafe(\hat A_{\ksafe}, \hat B_{\ksafe}) < 1.5 \Csafe(\Astar, \Bstar)$, we get
\[
    E = \mathcal{O}\!\left(\log d_s + \log \Csafe(\Astar, \Bstar)\right) = \mathcal{O}\!\left(\log d_s + \log \Pstarnorm\right).
\]
Each $z_i = \lfloor \hat\theta_i / \Delta \rceil$ with $\Delta = 2^{-E} \asymp \epsilon_{\mathrm{target}} / \sqrt{d_s}$, so Elias coding gives
\begin{equation}\label{eq:Binit}
    B_{\mathrm{init}} = \mathcal{O}\!\left(d_s \log\!\left(1 + \theta_{\max} \cdot \frac{\sqrt{d_s}}{\epsilon_{\mathrm{target}}}\right)\right),
\end{equation}
where $\theta_{\max} = \max_{i,j}\{|A_{ij}|, |B_{ij}|\}$. On $\Esafe$, $\hat\theta_{\ksafe}$ is within $\mathcal{O}(\epsilon_{\mathrm{target}})$ of $\theta_\star$, so this is $\mathcal{O}(d_s \log(1 + \theta_{\max}\, \Csafe \sqrt{d_s}))$, a finite one-time cost.

\medskip\noindent\textbf{Total.}
Combining all four components:
\begin{align}
    B(T) &= B_q(T) + B_m(T) + B_{\mathrm{flag}} + B_{\mathrm{init}} \nonumber \\
    &\le \left[d_s \log_2\!\left(1 + \frac{2}{\varrho}\right) + b_\varrho\right] K_T + B_{\mathrm{init}} + \mathcal{O}\!\left(\log \tau_{\ksafe} + \log^2 \tauls\right). \label{eq:total-bits-proof}
\end{align}
Since $K_T \le \lceil \log_2 T \rceil + 1 \le \log_2 T + 2$, this gives $B(T) = \mathcal{O}(d_s \log T)$.
\end{proof}

\section{Proof of Theorem~\ref{thm:main_ach}: Regret Bound}
\label{app:achievability}

\subsection{Proof of the Contraction Lemma}

\begin{proof}[Proof of Lemma~\ref{lem:contraction}]
    For $k > k_{\mathrm{ls}}$, the OLS bounds hold. Let $\mathrm{OLS}_k := \Vslowhat\,\tau_k^{-1/4} + \Vfasthat\,\tau_k^{-1/2}$. By the triangle inequality centered at $\theta_\star$ and non-expansiveness of projection:
    \[
        \|\Delta_k\|_2 \le \|\bar\theta_k - \theta_\star\|_2 + \|\theta_\star - \bar\theta_{k-1}\|_2 + \|\bar\theta_{k-1} - \tilde\theta_{k-1}\|_2 \le \mathrm{OLS}_k + \mathrm{OLS}_{k-1} + \varrho s_{k-1}.
    \]
    By the design of $\cslow$ and $\cfast$ in \eqref{eq:cslow-impl}--\eqref{eq:cfast-impl}:
    \[
        \mathrm{OLS}_k + \mathrm{OLS}_{k-1} \le \sbase - \varrho s_{k-1}^{\mathrm{base}}.
    \]
    Indeed, using $\tau_{k-1}^{-1/4} = 2^{1/4}\tau_k^{-1/4}$ and $\tau_{k-1}^{-1/2} = \sqrt{2}\,\tau_k^{-1/2}$, the left-hand side equals $\Vslowhat(1+2^{1/4})\tau_k^{-1/4} + \Vfasthat(1+\sqrt{2})\tau_k^{-1/2}$, while the right-hand side expands to $\cslow(1-\varrho\,2^{1/4})\tau_k^{-1/4} + \cfast(1-\varrho\sqrt{2})\tau_k^{-1/2}$, and by the definitions \eqref{eq:cslow-impl}--\eqref{eq:cfast-impl} these are identical.
    Substituting $s_{k-1} = m_{k-1} s_{k-1}^{\mathrm{base}}$:
    \[
        \|\Delta_k\|_2 \le \sbase + \varrho s_{k-1}^{\mathrm{base}}(m_{k-1} - 1).
    \]
    Dividing by $\sbase$ and using $s_{k-1}^{\mathrm{base}}/\sbase \le \sqrt{2}$ (since the ratio is a weighted average of $2^{1/4}$ and $\sqrt{2}$ with positive weights $\cslow\tau_k^{-1/4}$ and $\cfast\tau_k^{-1/2}$, hence bounded by $\sqrt{2}$):
    \[
        \frac{\|\Delta_k\|_2}{\sbase} \le 1 + \varrho\sqrt{2}(m_{k-1}-1).
    \]
    Since $m_k = \lceil \|\Delta_k\|_2/\sbase \rceil$ and $\lceil x \rceil < x + 1$, the contraction follows. With $\beta = \varrho\sqrt{2} < 1$, the sequence converges exponentially to the fixed point $m^\star = (2-\beta)/(1-\beta) = \mathcal{O}(1)$.
\end{proof}

\subsection{Proof of the Regret Convolution Bound}

\begin{lemma}[Regret Convolution Bound]
    \label{lem:convolution}
    For any non-negative sequence $U_k$ satisfying $s_k \le \sum_{j=\ksafe}^k \varrho^{k-j} U_j$, with $\varrho < 1/\sqrt{2}$,
    \begin{equation}
        \sum_k \tau_k s_k^2 \le \Crho^2 \sum_k \tau_k U_k^2, \qquad \text{where } \Crho = \frac{1}{1-\varrho\sqrt{2}}.
    \end{equation}
\end{lemma}

\noindent \textbf{Remark.} By the definition of $m_k$, $s_k \le \|\Delta_k\|_2 + \sbase$. By the triangle inequality, $\|\Delta_k\|_2 \le E_k + E_{k-1} + \varrho s_{k-1}$ where $E_k := \|\bar\theta_k - \theta_\star\|_2$. This gives the envelope recurrence $s_k \le \varrho s_{k-1} + U_k$ with $U_k := E_k + E_{k-1} + \sbase$.

\begin{proof}[Proof of Lemma~\ref{lem:convolution}]
    We split $\varrho^{k-j} = (\varrho\sqrt{2})^{(k-j)/2} \cdot (\varrho/\sqrt{2})^{(k-j)/2}$ and apply Cauchy--Schwarz inequality to $s_k \le \sum_j \varrho^{k-j} U_j$:
    \[
        s_k^2 \le \underbrace{\sum_{j \le k} (\varrho\sqrt{2})^{k-j}}_{\le\, \Crho} \cdot \sum_{j \le k} \left(\frac{\varrho}{\sqrt{2}}\right)^{\!k-j} U_j^2.
    \]
    Multiplying $s_k^2$ by $\tau_k = 2^k$ and summing over $k$, we obtain:
    \[
        \sum_k \tau_k s_k^2 \le \Crho \sum_k \sum_{j \le k} 2^k \left(\frac{\varrho}{\sqrt{2}}\right)^{\!k-j} U_j^2.
    \]
    Setting $m = k-j$ and exchanging the order of summation, the factor $2^k(\varrho/\sqrt{2})^{k-j} = 2^j \cdot (2 \cdot \varrho/\sqrt{2})^m = \tau_j\,(\varrho\sqrt{2})^m$, giving
    \[
        \sum_k \tau_k s_k^2 \le \Crho \sum_j \tau_j U_j^2 \sum_{m=0}^\infty (\varrho\sqrt{2})^m = \Crho^2 \sum_j \tau_j U_j^2.
    \]
\end{proof}

\subsection{Bound on Safety Event Probability}

\begin{lemma}[Lem. 5.5 from \cite{simchowitz20a}]
    The event $\mathcal{E}_\safe$ holds with probability at least $1 - \frac{\delta}{2}.$
\end{lemma}

\begin{proof}
    Let
    \begin{equation}
        \mathcal{E}_{\mathrm{conf}} = \bigcap_{k: \mathbf{\Lambda}_k \succeq I} \left\{ \left\| \hat A_k - A \mid \hat B_k - B \right\|_{\text{op}}^2 \leq \Conf_k \right \},
    \end{equation}
    where $\bigcap$ denotes the intersection of events over all $k$ such that $\mathbf{\Lambda}_k \succeq I$, and $\mathbf{\Lambda}_k$ is defined in~\eqref{eq:Lambda-k}.
    Lemma E.2 in \cite{simchowitz20a}, which is a Corollary of Theorem 1 from \cite{abbasi}, applied with empirical covariance $\mathbf{\Lambda}=\mathbf{\Lambda}_k$, baseline matrix $\Lambda_0 = I$, and confidence parameter $\delta/(4k^2)$, yields
    \begin{equation}
        \mathbb{P} \left[ \left\{ \left\| \hat A_k - A \mid \hat B_k - B \right\|_{\text{op}}^2 \geq \Conf_k \right\} \cap \left\{ \mathbf{\Lambda}_k \succeq I \right\} \right] \leq \frac{\delta}{4k^2}.
    \end{equation}
    Applying the union bound over all epochs and noting that at epoch $k_\safe$, $\mathbf{\Lambda}_{k_\safe} \succeq I$, we have that
    \begin{equation}
        \mathbb{P}[\mathcal{E}_\safe] \geq  \mathbb{P} \left[ \mathcal{E}_{\mathrm{conf}} \right] \geq  1 - \sum_{k=1}^\infty \frac{\delta}{4k^2} \geq 1 - \frac{\delta}{2}.
    \end{equation}
\end{proof}

\subsection{Proof of the Pre-Safe Regret Lemma}

\begin{proof}[Proof of Lemma~\ref{lem:presaferegret}]
    Before epoch $\ksafe$, Algorithm~\ref{alg:quantized_ce_lqr_proj} executes the baseline stabilizing controller $K_0$ with independent Gaussian exploration noise, identical to the unquantized algorithm in \cite{simchowitz20a}. By Lemma G.9 in \cite{simchowitz20a}, for $\delta < 1/T$, the cost during this phase is bounded with probability $1 - \delta$ by:
    \begin{equation}
        \label{presaferegret}
        \sum_{t=1}^{\tau_{\ksafe}-1} x_{t}^\top \Rx x_{t} + u_{t}^\top \Ru u_{t} \le d\,\tau_{\ksafe}\,\PsiB^2\, \Pzero \log\frac{1}{\delta}.
    \end{equation}

    What remains is to upper-bound the stopping time $\tau_{\ksafe}$. Under the Decoupled Statistical Trigger, the safe phase is declared as soon as $\mathbf{\Lambda}_k \succeq I$ \eqref{eq:Lambda-k} and:
    \begin{equation}
        \label{eq:decoupled-trigger}
        \sqrt{\Conf_k} \le \frac{1}{9\,\Csafe(\hat A_k, \hat B_k)}.
    \end{equation}

    We proceed by contraposition to find the maximum duration the algorithm can remain unsafe. Assume the statistical confidence is sufficiently small with respect to the \emph{true} system:
    \begin{equation}
        \label{eq:contra-assume}
        \sqrt{\Conf_k} \le \frac{2}{27\,\Csafe(A, B)}.
    \end{equation}
    Conditioned on $\Esafe$ \eqref{eq:Esafe}, assumption \eqref{eq:contra-assume} implies:
    \[
        \max\left\{\opnorm{\hat A_k - A},\;\opnorm{\hat B_k - B}\right\} \le \sqrt{\Conf_k} \le \frac{2}{27\,\Csafe(A, B)} < \frac{1}{3\,\Csafe(A, B)}.
    \]
    Lemma \ref{lem:11} (with center $(A, B)$) then gives $\Csafe(\hat A_k, \hat B_k) \le 1.5\,\Csafe(A, B)$.

    Substituting this bound into our algorithm's trigger threshold yields:
    \[
        \frac{1}{9\,\Csafe(\hat A_k, \hat B_k)} \ge \frac{1}{9 \times 1.5\,\Csafe(A, B)} = \frac{2}{27\,\Csafe(A, B)}.
    \]
    Therefore, if \eqref{eq:contra-assume} holds, we have $\sqrt{\Conf_k} \le \frac{1}{9\,\Csafe(\hat A_k, \hat B_k)}$, meaning the algorithm's trigger condition \eqref{eq:decoupled-trigger} strictly passes.

    By contraposition, if the algorithm has \emph{not} yet switched to the safe phase, it must be the case that \eqref{eq:contra-assume} is false:
    \[
        \Conf_k > \frac{4}{729\,\Csafe(A, B)^2}.
    \]
    By the definition of the safe constant, $\Csafe(A, B) \le 54\,\Pstarnorm^5$, yielding the deterministic necessary condition to remain in the pre-safe phase:
    \begin{equation}
        \label{eq:conf-lower-bound}
        \Conf_k > \frac{4}{729 \cdot 54^2\,\Pstarnorm^{10}} = \frac{1}{729^2 \,\Pstarnorm^{10}}.
    \end{equation}

    We now apply Lemma G.11 from \cite{simchowitz20a}, which provides a high-probability upper bound on $\Conf_k$ under the baseline controller $K_0$:
    \begin{equation}
        \label{eq:conf-upper-bound}
        \Conf_k \le \frac{d(1 + \opnorm{K_0}^2)}{\tau_k} \log\frac{\PsiB^2 \Jzero}{\delta}.
    \end{equation}
    Combining the necessary condition \eqref{eq:conf-lower-bound} with the upper bound \eqref{eq:conf-upper-bound}:
    \[
        \frac{1}{729^2\,\Pstarnorm^{10}} < \frac{d(1 + \opnorm{K_0}^2)}{\tau_k} \log\frac{\PsiB^2 \Jzero}{\delta}.
    \]
    Solving for $\tau_k$, the algorithm must have switched before:
    \[
        \tau_{\mathrm{stop}} \le 729^2\,\Pstarnorm^{10}\, d(1 + \opnorm{K_0}^2) \log\frac{\PsiB^2 \Jzero}{\delta}.
    \]
    Since epoch lengths double, $\tau_{\ksafe} \le 2\tau_{\mathrm{stop}}$. Substituting into \eqref{presaferegret} yields \eqref{presaferegretfinal}.
\end{proof}

\subsection{Proof of Theorem~\ref{thm:main_ach}}

\begin{proof}
    We analyze the post-safe regret using the decomposition \eqref{regredecomp}. Three terms require treatment: the per-epoch excess cost, the boundary term, and the lower-order remainder.

    \medskip\noindent\textbf{Per-epoch excess cost.}
    By Lemma~\ref{lem:51} point (1), for every post-safe epoch $k \ge \ksafe$,
    \begin{equation}\label{eq:Jk-bound}
        J_k - \Jstar \lesssim \Pstarnorm^8 \left(\Fnorm{\tilde A_k - \Astar}^2 + \Fnorm{\tilde B_k - \Bstar}^2\right).
    \end{equation}
    We split the sum over post-safe epochs $k \ge \ksafe$ into two ranges:
    \begin{itemize}
        \item \emph{Asymptotic epochs} ($\tauk > c\tauls$): the OLS rate~\eqref{OLSrate} applies, and these produce the main $\sqrt{T}$ contribution.
        \item \emph{Transient epochs} ($\ksafe \le k$ with $\tauk \le c\tauls$): the OLS rate does not yet hold, but both $(\tilde A_k, \tilde B_k)$ and $(A,B)$ lie in $\Bsafe$.
    \end{itemize}
    For the transient epochs, the safety derivation in Appendix~\ref{app:lemmas} gives
    \[
        \max\{\opnorm{\tilde A_k - \Astar},\, \opnorm{\tilde B_k - \Bstar}\} \le \frac{1}{\Csafe(\Astar,\Bstar)} = \mathcal{O}(1),
    \]
    so $J_k - \Jstar = \mathcal{O}(1)$ per epoch. The number of transient epochs is $\mathcal{O}(\log(\tauls/\tau_{\ksafe}))$ and each has $\tauk \le c\tauls = \mathcal{O}(\mathrm{poly}(\log T))$, so
    \[
        \sum_{k: \tauk \le c\tauls} \tauk (J_k - \Jstar) = \mathcal{O}(\mathrm{poly}(\log T)).
    \]

    \medskip\noindent\textbf{Epoch-boundary state norms.}
    By Lemma~5.3 from \cite{simchowitz20a}, on the event $\Esafe \cap \Ebound$,
    \[
        \|x_{\tauk}\|_2^2 \lesssim \PsiB \Jzero \log\frac{1}{\delta}\, \Pstarnorm^3, \qquad \forall k \ge \ksafe.
    \]
    Substituting into the boundary term of~\eqref{regredecomp}:
    \[
        \log T \max_{k \le \log T} \|x_{\tauk}\|_2^2 \lesssim \PsiB \Jzero \Pstarnorm^3 \log T \cdot \log\frac{1}{\delta} \le \PsiB \Jzero \Pstarnorm^3 \tauls \log\frac{1}{\delta},
    \]
    where the last inequality uses $\log T \le \tauls$ for $T$ sufficiently large.

    \medskip\noindent\textbf{Combining.}
    Absorbing the transient-epoch and boundary contributions into the lower-order terms, and retaining only the asymptotic sum, we arrive at
    \begin{align}
         & \sum_{k=k_\safe}^{k_{\mathrm{fin}}} \tau_k(J_k - \Jstar) + \log T \max_{k \leq \log T} \lVert x_{\tau_k}\rVert_2^2 \nonumber \\
         & \lesssim \sum_{k: \tauk > \tauls} \tau_k \Pstarnorm^8 \left( \Fnorm{\Atil_k - \Astar}^2 + \Fnorm{\Btil_k - \Bstar}^2 \right) + \PsiB \Jzero \Pstarnorm^3 \tauls\log \frac{1}{\delta}. \label{eq:post-safe-split}
    \end{align}

    Focusing on the first term, we separate the parameter error into the unquantized OLS error and the quantization error. The regret depends on the \emph{controller-side} projection $\check\theta_k = \Pi_{\Bsafe}(\tilde\theta_k)$~\eqref{eq:safe-proj-overview}. Since $\theta_\star \in \Bsafe$ on $\Esafe$, non-expansiveness of projection gives $\Fnorm{\check\theta_k - \theta_\star} \le \Fnorm{\tilde\theta_k - \theta_\star}$, so it suffices to bound the decoded error $\Fnorm{\tilde\theta_k - \theta_\star}$.

    For the plant, the projected estimate is $\bar\theta_k = \Pi_{\Bsafe}(\hat\theta_k)$. By identical reasoning, $\Fnorm{\bar\theta_k - \theta_\star} \le \Fnorm{\hat\theta_k - \theta_\star}$. Furthermore, because the quantized message bounds the reconstruction error by the scale radius, $\Fnorm{\tilde\theta_k - \bar\theta_k} \le \varrho s_k$, so applying the triangle inequality and $(a+b)^2 \le 2a^2 + 2b^2$ gives:
    \begin{align}
        & \sum_{k:\tau_k > \tauls} \tau_k \Pstarnorm^8 \left( \Fnorm{\Atil_k - \Astar}^2 + \Fnorm{\Btil_k - \Bstar}^2 \right) \nonumber \\
        & \le 2\Pstarnorm^8 \sum_{k:\tau_k > \tauls} \tau_k \left( \Fnorm{\Ahat_k - \Astar}^2 + \Fnorm{\Bhat_k - \Bstar}^2 + \varrho^2 s_k^2 \right). \label{eq:quant-error-split}
    \end{align}

    Conditioned on the high-probability event $\Els \cap \Esafe \cap \Ebound$, the OLS errors are deterministically bounded. By the definition of the multiplier \eqref{eq:adaptive-mult}, $m_k = \max\{1, \lceil\|\Delta_k\|_2 / \sbase\rceil\} \le \|\Delta_k\|_2 / \sbase + 1$. Thus, $s_k \le \|\Delta_k\|_2 + \sbase$.

    By the triangle inequality, $\|\Delta_k\|_2 \le \Fnorm{\bar\theta_k - \theta_\star} + \Fnorm{\bar\theta_{k-1} - \theta_\star} + \Fnorm{\bar\theta_{k-1} - \tilde\theta_{k-1}}$. Let $E_k := \Fnorm{\bar\theta_k - \theta_\star}$. Because the reconstruction error is bounded by $\varrho s_{k-1}$, we obtain the recurrence for all tracking epochs $k > \ksafe$:
    \begin{equation}
        s_k \le \varrho s_{k-1} + \underbrace{E_k + E_{k-1} + \sbase}_{:= U_k}.
    \end{equation}

    To unroll this recurrence, we must account for the base case at epoch $\ksafe$. The differential tracking loop begins at $\ksafe+1$; at $\ksafe$ itself, \textsc{AbsoluteInit} returns $\tilde\theta_{\ksafe}$. On $\Esafe$, both $\tilde\theta_{\ksafe}$ and $\theta_\star$ lie in $\Bsafe$, so by the safety derivation in~\eqref{eq:safety_derivation},
    \begin{equation}\label{eq:sinit-def}
        \sinit := \Fnorm{\tilde\theta_{\ksafe} - \theta_\star} \le \frac{\sqrt{d_s}}{\Csafe(\Astar,\Bstar)} = \mathcal{O}(1),
    \end{equation}
    where $d_s = \dx(\dx+\du)$ and the $\sqrt{d_s}$ factor converts from the per-block operator norms to the vectorized Frobenius norm.
    Unrolling the recurrence down to this boundary gives
    \begin{equation}
        s_k \le \sum_{j=\ksafe+1}^k \varrho^{k-j} U_j \quad + \quad \varrho^{k-\ksafe} \sinit.
    \end{equation}

    To express this as a single convolution, we define $U_{\ksafe} := \sinit$ and write
    \begin{equation}
        \label{eq:envelope}
        s_k \le \sum_{j=\ksafe}^k \varrho^{k-j} U_j.
    \end{equation}

    By applying Lemma~\ref{lem:convolution} (since $\tau_k = 2^k$ and $\varrho < 1/\sqrt{2}$), the weighted $\ell_2$-norm of this convolution is strictly bounded:
    \begin{equation}\label{eq:Rquant-def}
        R_{\mathrm{quant}} := \Pstarnorm^8 \sum_{k=\ksafe}^{\log_2 T} \tau_k \varrho^2 s_k^2 \le \Pstarnorm^8 \varrho^2 \Crho^2 \sum_{k=\ksafe}^{\log_2 T} \tau_k U_k^2.
    \end{equation}

    We analyze $U_k$ on the good event across two phases.

    \medskip\noindent\textbf{Transient Phase} ($\tauk \le c\tauls$).
    For the base case $j=\ksafe$, $U_{\ksafe} = \sinit = \mathcal{O}(1)$ by~\eqref{eq:sinit-def}. For $j > \ksafe$, any projected parameter and the true system both lie in $\Bsafe$, so by the safety derivation (Appendix~\ref{app:lemmas}),
    \[
        E_j = \Fnorm{\bar\theta_j - \theta_\star} \le \frac{\sqrt{d_s}}{\Csafe(\Astar,\Bstar)} = \mathcal{O}(1), \qquad U_j = \mathcal{O}(1).
    \]

    \medskip\noindent\textbf{Asymptotic Phase} ($\tauk > c\tauls$).
    The statistical OLS bounds apply deterministically:
    \[
        E_k \le \mathrm{OLS}_k := \Vslowhat\,\tau_k^{-1/4} + \Vfasthat\,\tau_k^{-1/2}.
    \]
    By the design of the ideal base schedule (see~\eqref{eq:cslow-impl}--\eqref{eq:cfast-impl}) and the same calculation used in the proof of Lemma~\ref{lem:contraction}, $\mathrm{OLS}_k + \mathrm{OLS}_{k-1} \le \sbase$, so $U_k \le 2\sbase$.

    \medskip\noindent\textbf{Combining.}
    Substituting these bounds into~\eqref{eq:Rquant-def} gives the explicit split
    \begin{equation}\label{eq:Uk2-split}
        \sum_{k=\ksafe}^{\log_2 T} \tau_k U_k^2
        \;\le\;
        \underbrace{\sum_{\tauk \le c\tauls} \tau_k \cdot \mathcal{O}(1)}_{\le\; \mathcal{O}(\tauls)}
        \;+\;
        \underbrace{\sum_{\tauk > c\tauls} \tau_k \cdot 4(\sbase)^2}_{\text{asymptotic contribution}}.
    \end{equation}
    Expanding \eqref{eq:two-scale-schedule}, $(\sbase)^2 = (\cslow\tau_k^{-1/4} + \cfast\tau_k^{-1/2})^2 \le 2\cslow^2\tau_k^{-1/2} + 2\cfast^2\tau_k^{-1}$ and multiplying by $\tau_k$ yields terms $\tau_k^{1/2}$ and $\tau_k^0$. Since $\tau_k = 2^k$, the doubling-epoch geometric sums evaluate as
    \begin{equation}\label{eq:doubling-sums}
        \sum_{k=0}^{\lfloor \log_2 T \rfloor} \tau_k^{1/2} = \sum_{k=0}^{\lfloor \log_2 T \rfloor} 2^{k/2} \le \frac{2^{(\lfloor \log_2 T \rfloor+1)/2}-1}{\sqrt{2}-1} \lesssim \sqrt{T}, \qquad
        \sum_{k=0}^{\lfloor \log_2 T \rfloor} 1 = \lfloor \log_2 T \rfloor + 1 \lesssim \log T.
    \end{equation}
    Therefore the quantization contribution splits into two scales:
    \begin{align}
        R_{\mathrm{quant}}
         &\lesssim \Pstarnorm^8 (\varrho \Crho \cslow)^2 \sqrt{T} + \Pstarnorm^8 (\varrho \Crho \cfast)^2 \log T + \Pstarnorm^8 \varrho^2 \Crho^2 \cdot \mathcal{O}(\tauls). \label{eq:two-scale-regret-decomp}
    \end{align}

    \medskip\noindent\textbf{Full regret bound.}
    Combining~\eqref{eq:quant-error-split} with the OLS error from~\eqref{OLSrate}, we obtain:
    \begin{align}
        \eqref{eq:quant-error-split} &\lesssim \sqrt T \left(\frac{\dx\du}{\sigmain^2} \Pstarnorm^{10} \log \frac{1}{\delta} + \Pstarnorm^{8} (\varrho \Crho \cslow)^2 \right) + \log T \cdot \Pstarnorm^{8} (\varrho \Crho \cfast)^2 + \dx^2 \Pstarnorm^3 \tauls \log^2\frac{1}{\delta}.
    \end{align}

    Further combining this with the regret decomposition~\eqref{regredecomp}, we have:
    \begin{align}
        \begin{split}
            \sum_{t=\tau_{\ksafe}}^T (x_t^\top \Rx x_t + u_t^\top \Ru u_t - \Jstar)
             & \lesssim \sqrt T \bigg( \du \sigmain^2 \PsiB^2 \Pstarnorm + \sqrt{d \log \frac{1}{\delta}} \Pstarnorm^4            \\
             & \qquad + \frac{\dx\du}{\sigmain^2} \Pstarnorm^{10} \log \frac{1}{\delta} + \Pstarnorm^{8}(\varrho \Crho \cslow)^2 \bigg)       \\
             & \quad + \log T \cdot \Pstarnorm^{8}(\varrho \Crho \cfast)^2 \\
             & \quad + ( 1 + \sqrt d \sigmain^2 \PsiB) \Pstarnorm^4 \log^2 \frac{1}{\delta} + \PsiB \Jzero \Pstarnorm^3 \tauls \log \frac{1}{\delta} \\
             & \quad + \dx^2 \Pstarnorm^3 \tauls \log ^2 \frac{1}{\delta}
        \end{split}
    \end{align}
    Substituting $\sigmain^2 \asymp \sqrt{\dx}\,\Pstarnorm^{9/2}\,\PsiB\,\sqrt{\log(\Pstarnorm/\delta)}$ from Lemma~\ref{lem:51} point~(6) into the $\sigmain^2$-dependent terms:
    \begin{itemize}
        \item $\du\sigmain^2\PsiB^2\Pstarnorm = \du\sqrt{\dx}\,\PsiB^3\,\Pstarnorm^{11/2}\,\sqrt{\log(\Pstarnorm/\delta)}$.
        \item $\frac{\dx\du}{\sigmain^2}\Pstarnorm^{10}\log\frac{1}{\delta} = \frac{\du\sqrt{\dx}\,\Pstarnorm^{11/2}\,\log(1/\delta)}{\PsiB\,\sqrt{\log(\Pstarnorm/\delta)}} \le \frac{\du\sqrt{\dx}\,\Pstarnorm^{11/2}\,\sqrt{\log(\Pstarnorm/\delta)}}{\PsiB}$,
        which is at most the first term since $\PsiB \ge 1$.
        \item $\sqrt{d}\,\sigmain^2\,\PsiB\,\Pstarnorm^4\,\log^2\frac{1}{\delta} = \sqrt{d\dx}\,\PsiB^2\,\Pstarnorm^{17/2}\,\sqrt{\log(\Pstarnorm/\delta)}\,\log^2\frac{1}{\delta}$.
    \end{itemize}
    Substituting the definitions of $\cslow$~\eqref{eq:cslow-impl} and $\cfast$~\eqref{eq:cfast-impl} and the proxy comparison~\eqref{eq:Pophat-ratio} recovers $Q_{\mathrm{slow}}(\varrho)$~\eqref{eq:Qslow} and $Q_{\mathrm{fast}}(\varrho)$~\eqref{eq:Qfast} in the quantization terms, yielding Theorem~\ref{thm:main_ach}(b).
\end{proof}

\end{document}